\documentclass[review]{elsarticle}
\usepackage{amsmath}	
\usepackage{amssymb}
\usepackage{amsthm}
\usepackage{pdfpages}
\newtheorem{proposition}{Proposition}
\newtheorem{remark}{Remark}

\usepackage{lineno,hyperref}
\usepackage{subfig}
\usepackage{color}
\modulolinenumbers[5]
\journal{Robotics and Autonomous Systems}
\begin{document}

\begin{frontmatter}

\title{Flatness-based nonlinear control strategies for trajectory tracking of quadcopter systems}

\author[LCIS]{Thinh Nguyen}
\ead{ngoc-thinh.nguyen@lcis.grenoble-inp.fr}
\author[LCIS]{Ionela Prodan}
\ead{ionela.prodan@lcis.grenoble-inp.fr}
\author[LCIS]{Laurent Lef\`evre}
\ead{laurent.lefevre@lcis.grenoble-inp.fr}
\address[LCIS]{Univ. Grenoble Alpes, LCIS, F-26902, Valence, France}

\begin{abstract}
This paper proposes several nonlinear control strategies for trajectory tracking of a quadcopter system based on the property of differential flatness. Its originality is twofold. Firstly, it provides a flat output for the quadcopter dynamics capable of creating full flat parametrization of the states and inputs. Moreover, B-splines characterizations of the flat output and their properties allow for optimal trajectory generation subject to way-point constraints. Secondly, several control strategies based on computed torque control and feedback linearization are presented and compared. The advantages of flatness within each control strategy are analyzed and detailed through extensive simulation results.
\end{abstract}


\end{frontmatter}
\linenumbers
\section{Introduction}
Recently, there has been an increasing interest in multiple research communities for the Unmanned Aerial Vehicles (UAVs) investigating on kinematics and dynamics, trajectory generation, guidance, navigation and control, especially for quadcopters \cite{dong2016high,ha2014passivity,formentin2011flatness,khan2014quadcopter,sydney2013dynamic,sa2011estimation}. The quadcopters seem to become popular only in the last decades but in fact, their concepts appeared more than a century ago. The first prototype, which was built in 1907 and named Brequet-Richet Gyrolane No.1, is reported to have lifted into flight \cite{leishman2001breguet}. Nowadays, quadcopters are being widely used in different domains and for many purposes such as research platform \cite{dong2016high,chovancova2016comparison,formentin2011flatness,khan2014quadcopter,sydney2013dynamic,sa2011estimation}, military enforcement \cite{markman2000bell}, commercial use \cite{Phantom} as well as being in concept for medical emergency \cite{AmbulanceDrone}. 

For the research area, quadcopters are challenging vehicles to control as they are not only strongly nonlinear and underactuated but also subject to many operating constraints. One feasible approach is to generate off-line a reference path that allows tracking of specific objectives (i.e., passing through a priori given way-points, consumption minimization, state/input constraints satisfaction). Then, develop an effective tracking mechanism for the quadcopter to follow the reference at run-time \cite{iprodanCEP}. As a result, generating a trajectory which respects the internal dynamics of the system and various external constraints, becomes part of the problem.

A popular solution for trajectory generation is the use of flat output characterizations \cite{fliess1995flatness}. These allow implicitly to validate the dynamics and may (with some difficulty) take into account constraints. There is a number of works like \cite{formentin2011flatness,sreenath2013geometric,chamseddine2012flatness,sreenath2013dynamics} which employ differential flatness within the trajectory tracking control design. However, these approaches are lacking in several essential directions:
\begin{itemize}
\item simplified dynamics (usually the yaw angle and/or the thrust are kept constant) are used to generate the trajectory and hence tracking errors may ensue;
\item part of the available information provided by the trajectory is discarded at runtime (e.g., only position information is taken into account).
\end{itemize}

From the control point of view, there are various quadcopter control strategies like Lyapunov-based control \cite{bouabdallah2004design}, classical PID control \cite{argentim2013pid,kodgirwar2014design}, LQR (Linear-quadratic regulator) control \cite{chovancova2016comparison,argentim2013pid}, feedback linearization \cite{formentin2011flatness} or optimization-based control \cite{chamseddine2012flatness}. Each of these approaches has some significant shortcomings:
\begin{itemize}
\item the control mechanism considers only altitude and attitude components and discards the rest of the state components \cite{bouabdallah2004design}. As a remark, Lyapunov function and corresponding stabilization controller \cite{tzafestas2013introduction,bouabdallah2004design} may be difficult to find in other specific cases (e.g., controlling the position and direction angle of the quadcopter system); 
\item PID or LQR controllers (which are designed for a certain linearized model) are used to close the loop for the strongly non-linear dynamics of the quadcopter; this limits the performances of the scheme and requires for stay around the equilibrium point along which the linearization has been done \cite{argentim2013pid,kodgirwar2014design};
\item even when nonlinear dynamics are taken into account, simplifications and approximations are made (e.g., constant yaw angle \cite{rivera2010flatness,formentin2011flatness}, small angles \cite{sydney2013dynamic}, constant velocity \cite{iprodanCEP}).
\end{itemize}

These simplifications for both trajectory generation and tracking mechanisms are apologized by the inherent complexity of the quadcopter dynamics but they raise two questions: \emph{how can we make use of the full information provided by flatness? and, is it possible to control the quadcopter system considering its full behavior?}.  

To overcome the difficulties in processing the nonlinearities of a quadcopter system and all the above mentioned shortcomings, we propose in the rest of the paper several contributions which, to the best of our knowledge, are new to the state of the art:
\begin{itemize}
\item construct a flat trajectory which provides positions, angles, thrust and torques, considering the nonlinear quadcopter dynamics (throughout the paper we use B-splines characterizations of the flat output and their properties, which allow for optimal trajectory generation subject to way-point constraints \cite{obstacle2015MED});
\item delve into several control strategies based on the concept of feedback linearization which can control both orientation and position of the quadcopter system without any assumptions or simplifications on the system (as the assumptions of the nullified yaw angle \cite{rivera2010flatness,formentin2011flatness} or small angles \cite{sydney2013dynamic}).
\end{itemize}

The remaining paper is organized as follows. Section \ref{sec:model} presents and in-depth view of the  kinematics and dynamics modeling of a quadcopter system. Section \ref{sec:flat} presents the flatness-based quadcopter characterization which fully takes into account the system dynamics. Section \ref{sec:control} details some effective constructions for the rotation and attitude controllers of a quadcopter system based on feedback linearization. These constructions are further used to develop trajectory tracking control strategies making fully use of the information provided by flatness. Extensive simulation results and comparisons between the proposed control strategies are provided in Section \ref{sec:sim} over a Crazyflie quadcopter system. Section \ref{sec:concl} presents the conclusions and future work.
\section{Quadcopter modelling}
\label{sec:model}
This section introduces the kinematics and the associated dynamics of the quadcopter using Newton-Euler formalism (more information can be found in \cite{formentin2011flatness,khan2014quadcopter}). The quadcopter will operate in two different coordinate systems: the \emph{body reference frame} (BF) which is attached to the mass center of the quadcopter and the \emph{inertial reference frame} (IF) which is fixed to the ground (East-North-Up coordinates). Upper-scripts $B$ and $I$ will be used to denote a variable measured in the BF and in the IF, respectively. 
\subsection{Kinematics}
The angular position (or attitude) of the quadcopter is defined by the orientation of the BF with respect to the IF. In general, this relation is described through a 3D rotation matrix which is the product of the sequence of three successive rotations. For the quadcopter we apply the roll--pitch--yaw XYZ $(\phi, \theta, \psi)$ sequence whose rotation matrix is \footnote{Note that, in order to write in a more compact way we have used in this paper $'s'$, $'c'$ and $'t'$ to denote the $\sin(\cdot)$, $\cos(\cdot)$ and $\tan(\cdot)$ functions, respectively.}(similar results can be found in \cite{formentin2011flatness,sydney2013dynamic}): 
\begin{equation}
\label{eq:roma}
_{B}^{I}R=R_Z (\psi)R_Y (\theta)R_X (\phi)=\begin{bmatrix} c\theta c\psi & s\phi s\theta c\psi-c\phi s\psi & c\phi s\theta c\psi+s\phi s\psi \\ c\theta s\psi & s\phi s\theta s\psi + c\phi c\psi & c\phi s\theta s\psi-s\phi c\psi \\ -s\theta & s\phi c\theta & c\phi c\theta \end{bmatrix},
\end{equation}
The quadcopter has the angular velocity vector $\overrightarrow{\omega}$ pointing along the axis of rotation. We use the right hand rule to determine the direction of the rotation corresponding to the one of the angular velocity vector. Therefore, the angular velocity vector $\overrightarrow{\omega}$ looked from the BF $^{B}\overrightarrow{\omega} \triangleq \begin{bmatrix} \omega_x \ \omega_y \ \omega_z \end{bmatrix}^\top$ \footnote{Note that the angular velocity $^{B}\overrightarrow{\omega}$ is physically measured by the gyroscope.} can be expressed in term of the attitude as (the inverse relation can be found in \cite{formentin2011flatness,sydney2013dynamic}):
\begin{equation}
\label{eq:omega}
^{B}\overrightarrow{\omega}=\begin{bmatrix} 1 & 0 & -s\theta \\ 0 & c\phi & s\phi c\theta \\ 0 & -s\phi & c\phi c\theta \end{bmatrix} \begin{bmatrix} \dot{\phi} \\ \dot{\theta} \\ \dot{\psi} \end{bmatrix}=W \dot{\eta}, 
\end{equation}
where $\eta \triangleq \begin{bmatrix} \phi \ \theta \ \psi \end{bmatrix}^\top$.
\subsection{Dynamics}
The quadcopter structure and the BF are presented in Figure \ref{fig:Quad} including the corresponding angular velocities  $\omega_i$, torques $\tau_i$ and forces $f_i$ created by the four rotors, with $i=1,\cdots,4$.
\begin{figure}
\begin{center}
\includegraphics[width=0.5\columnwidth]{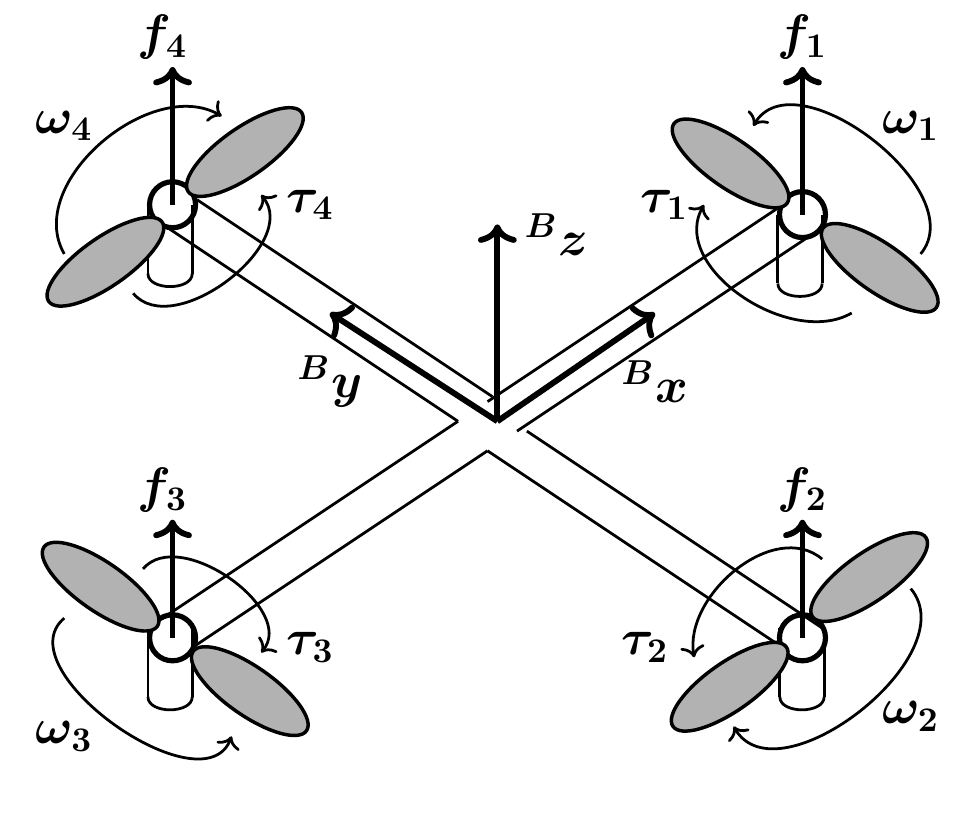}
\caption{Quadcopter system.}
\label{fig:Quad}
\end{center}
\end{figure}

From the aerodynamic effects viewpoint, we can express the torque $\tau_{i}$ about the $z_B$ axis\footnote{$(-1)^i b\omega_i^2$ term is positive if the $i^{th}$ propeller is spinning clockwise and negative if it is spinning counterclockwise} and the forces $f_i$ along $z_B$ direction for the $i^{th}$ rotor as:
\begin{align}
&\tau_{i}=(-1)^i b\omega_i^2 + I_M \dot{\omega_i} \approx (-1)^i b\omega_i^2, \\
&f_i=K_T \omega_i^2,
\end{align}
where $i=1,\cdots,4$, $I_M$ is the moment of inertia of the motor about the rotational axis, $b$ and $K_T$ are assumed known aerodynamic constants.\\
Furthermore, the total thrust force and torques acting on the quadcopter have the magnitudes as:
\begin{align}
\label{thrust}
&T=\sum_{i=1}^{4} f_i=K_T\sum_{i=1}^{4}\omega_i^2, \\
\label{tphi}
&\tau_{\phi}=Lf_4-Lf_2=LK_T\left(\omega_4^2 - \omega_2^2 \right), \\
\label{ttheta}
&\tau_{\theta}=Lf_3-Lf_1=LK_T\left(\omega_3^2  - \omega_1^2 \right), \\
\label{tpsi}
&\tau_{\psi}=\sum_{i=1}^{4} \tau_i=b\left( -\omega_1^2 + \omega_2^2 -\omega_3^2 +\omega_4^2 \right),
\end{align}
where $L$ is the distance from the center of the quadcopter to any propellers.
Note that, expressing in BF, the thrust force is defined as $\overrightarrow{^B T} \triangleq \begin{bmatrix} 0 \ 0 \ T \end{bmatrix}^\top$ and $\overrightarrow{\tau_\phi},\overrightarrow{\tau_\theta},\overrightarrow{\tau_\psi}$ have corresponding directions along the three axes of BF.
\subsubsection{Translation equation}
\hfill \newline
In the IF, assuming that the centrifugal force is nullified, hence, only gravitational force, $m \overrightarrow{g}$, thrust force, $\overrightarrow{^B T}$ and external perturbation force (most commonly, friction), $\overrightarrow{F_D}$ are contributing to the acceleration of the quadcopter:
\begin{equation}
\label{eq:tNewton}
m \ddot{\xi}=m \overrightarrow{g} +(_{B}^{I}R) \overrightarrow{^B T}+\overrightarrow{F_D},
\end{equation}
where $\xi \triangleq \begin{bmatrix} x \ y \ z \end{bmatrix}^\top$ represents the quadcopter position, the thrust force $\overrightarrow{^B T}$ has the magnitude defined in equation \eqref{thrust} and the perturbation force $\overrightarrow{F_D}$ will be detailed later in Section \ref{friction}.
\subsubsection{Rotation equation}
\hfill \newline
While it looks convenient to have the linear equations of motion in the IF, the rotational equations are more useful in the BF. We assume that the quadcopter has a symmetrical construction, hence, the inertial tensor $^BI$ is a diagonal matrix:
\begin{equation}
\label{eq:tensor}
^BI=diag \lbrace I_{xx}, I_{yy},I_{zz} \rbrace.
\end{equation}
In vector form, the Newton-Euler rotational equation for the quadcopter in BF taking into account the gyroscopic force is defined as:
\begin{equation}
\label{eq:rNewton}
^BI ^B\dot{\overrightarrow{\omega}}+^B\overrightarrow{\omega}\times(^BI ^B\overrightarrow{\omega})=\tau_\eta,
\end{equation}
where `$\times$' denotes the cross-product of two vectors and $\tau_\eta\triangleq \begin{bmatrix} \tau_{\phi} \  \tau_{\theta} \  \tau_{\psi} \end{bmatrix}^\top$ gathers the roll, pitch and yaw torques which have been detailed in equations \eqref{tphi}--\eqref{tpsi}.
\subsubsection{Perturbation force}
\label{friction}
\hfill \newline
In order to make the model more realistic and able to take into account air disturbances, we model the external perturbation force triggered by the quadcopter motion and the external wind. Based on the definition of friction force found in \cite{fox1994introduction}, the vector of global friction force is given by:
\begin{equation}
\label{perturbation}
\overrightarrow{F_D}=\frac{1}{2} C_D \rho |\overrightarrow{V_r}| A \overrightarrow{V_r},
\end{equation}
where $\rho$ is the surrounding fluid density, $C_D$ is the drag coefficient, $\overrightarrow{V_r}=\overrightarrow{w}-\dot{\xi}$ is the vector of relative motion between the wind speed $\overrightarrow{w}$ and the quadcopter velocity $\dot{\xi}$. In equation \eqref{perturbation}, the projected area $A$ is calculated by the following relation:
\begin{equation}
\label{projected}
A=A_x \begin{vmatrix} \frac{\overrightarrow{^I x_B}\overrightarrow{V_r}}{|\overrightarrow{V_r}|} \end{vmatrix} + A_y \begin{vmatrix} \frac{\overrightarrow{^I y_B}\overrightarrow{V_r}}{|\overrightarrow{V_r}|} \end{vmatrix} + A_z \begin{vmatrix} \frac{\overrightarrow{^I z_B}\overrightarrow{V_r}}{|\overrightarrow{V_r}|} \end{vmatrix},
\end{equation}
where $A_x$, $A_y$, $A_z$, which depend on the designed structure of the quadcopter, describe the projected areas into YZ, XZ, and XY planes of the BF. In equation \eqref{projected}, $\overrightarrow{^I x_B}$, $\overrightarrow{^I y_B}$ and $\overrightarrow{^I z_B}$ represent the three column vectors of the rotation matrix $^I _B R$ given in \eqref{eq:roma}.
\section{Flat characterizations}
\label{sec:flat}
This section introduces first some basic definitions and notions on differential flatness and B-splines parametrization \cite{obstacle2015MED,iprodanCEP}. Next, a novel flatness-based characterization which fully takes into account the dynamics of the quadcopter system is described. 
\subsection{Basic definitions}
\label{subsec:def}

Differential flatness represents a generalization to nonlinear systems of the structural properties of the linear systems, which exhibit a state representation obtained via derivatives of the input and output signals.

Consider a general system: 
\begin{equation}
\label{eq:QuadG}
\dot{\bold{x}}(t)=f(\bold{x}(t),\bold{u}(t)),
\end{equation}
where $x(t) \in \mathbb{R}^{n}$ is the state vector and $u(t) \in \mathbb{R}^{m}$ is the input vector. The nonlinear system written in general form as in equation \ref{eq:QuadG} is called differentially flat if there exists a flat output $\bold{z}(t) \in \mathbb{R}^{m}$:
\begin{equation}
\label{eq:flat_out}
\bold{z}(t)=\Upsilon(\bold{x}(t),\bold{u}(t),\dot{\bold{u}}(t),\cdots, \bold{u}^{(q)}(t)),
\end{equation}
such that the states and inputs can be algebraically expressed in terms of $\mathbf z(t)$ and a finite number of its higher-order derivatives:
\begin{subequations}
\label{eq:diff_a}
\begin{align}
\label{eq:diff_a_1}\bold{x}(t)&=\Upsilon_1(\bold{z}(t),\dot{\bold{z}}(t), \cdots ,\bold{z}^{(q)}(t)),\\
\label{eq:diff_a_2}\bold{u}(t)&=\Upsilon_2(\bold{z}(t),\dot{\bold{z}}(t),\cdots ,\bold{z}^{(q+1)}(t)).
\end{align}
\end{subequations}

\begin{remark} []
\label{rem:flatInput}
Note that the flatness and controlability properties of a system are directly related. It is demonstrated that a linear system is flat if and only if it is controllable \cite{levine2011necessary}\cite{fliess1995flatness}. Furthermore, for any system admitting a flatness-based representation, the number of flat outputs equals the number of inputs. \ensuremath{\hfill\square}
\end{remark}
An essential aspect of construction \eqref{eq:flat_out}--\eqref{eq:diff_a} is that it reduces the problem of trajectory generation to finding an adequate flat output \eqref{eq:flat_out}. This means choosing $\mathbf z(t)$ such that, via mappings $\Upsilon_1,\Upsilon_2$, various constraints on states and inputs \eqref{eq:diff_a} are verified. Since the flat output is not straightforward to compute under these restrictions, usually a projection across a finite basis of smooth functions $\Lambda^i(t)$ is considered:
\begin{equation}
\label{eq:paramFlat}
\bold{z}(t)=\sum\limits_{i=1}^{n}{\alpha_i \Lambda^i(t)}, \ \  \alpha_i \in \mathbb{R}.
\end{equation}
Parameter $n$ in equation \eqref{eq:paramFlat} depends on the number of constraints imposed onto the system \cite{wilkinson}.

There are multiple choices for the basis functions $\Lambda^i(t)$. Among these, \emph{B-spline} basis functions are well-suited to flatness parametrization due to their ease of enforcing continuity and because their degree depends only up to which derivative is needed to ensure continuity \cite{suryawan2012constrained,de2009flatness}.

A B-spline of order $d$ is characterized by a \emph{knot-vector} \cite{piegl1995b}: 
\begin{equation}
\label{eq:knot_vector}
\mathbb T=\left\{\tau_0,\tau_1\dots  \tau_m\right\},
\end{equation}
of non-decreasing time instants ($\tau_0\leq \tau_1\leq\dots \leq \tau_m$) which parametrizes the associated basis functions $B_{i,d}(t)$:
\begin{subequations}
\begin{align}
\label{eq:bsplines_0}
B_{i,1}(t)&=\begin{cases}1, \textrm{ for }\tau_i\leq t<\tau_{i+1}\\0\textrm{ otherwise}\end{cases},\\
\label{eq:bsplines_d}
B_{i,d}(t)&=\frac{t-\tau_i}{\tau_{i+d-1}-\tau_i}B_{i,d-1}(t)+\frac{\tau_{i+d}-t}{\tau_{i+d}-\tau_{i+1}}B_{i+1,d-1}(t),
\end{align}
\end{subequations}
for $d>1$ and $i=0,1\dots n=m-d$.
Considering a collection of \emph{control points} 
\begin{equation}
\label{eq:control_points}
\mathbb P=\left\{p_0,p_1\dots  p_n\right\},
\end{equation}
we define a \emph{B-spline curve} as a linear combination of the control points \eqref{eq:control_points} and the B-spline functions \eqref{eq:bsplines_0}--\eqref{eq:bsplines_d}:
\begin{equation}
\label{eq:bspline_curve}
\bold{z}(t)=\sum\limits_{i=0}^n B_{i,d}(t)p_i=\mathbf P\mathbf B_d(t),
\end{equation}
where $\mathbf P=\begin{bmatrix} p_0\dots  p_n\end{bmatrix}$ and $\mathbf B_d(t)=\begin{bmatrix} B_{0,d}(t)\dots  B_{n,d}(t)\end{bmatrix}^\top$.
This construction yields several interesting properties which are enumerated in \cite{obstacle2015MED}.

Let us consider now a collection of $N+1$ way-points and the time stamps associated to them:
\begin{equation}
\label{eq:wt}
\mathbb W=\{w_k\}\textrm{ and } \mathbb T_{\mathbb W}=\{t_k\},
\end{equation}
for any $k=0\dots N$. The goal is to construct a flat trajectory which passes through each way-point $w_k$ at the time instant $t_k$, i.e., find a flat output $\bold{z}(t)$ such that 
\begin{equation}
\label{eq:xconstr}
\bold{x}(t_k)=\Upsilon_1(\bold{z}(t_k),\dots \bold{z}^{(q)}(t_k))=w_k, \: \forall k=0\dots N.
\end{equation}

Within the B-spline framework \eqref{eq:bspline_curve} we provide a vector of control points \eqref{eq:control_points} and its associated knot-vector \eqref{eq:knot_vector} such that condition \eqref{eq:xconstr} is verified:
\begin{equation}
\label{eq:flatconstraints}
\Upsilon_1(\mathbf B_d(t_k),\mathbf P)=w_k,\: \forall k=0\dots N.
\end{equation}

Let us assume that the knot-vector is fixed ($\tau_0=t_0$, $\tau_{n+d}=t_N$ and the intermediary points $\tau_{d},\dots, \tau_{n}$ are equally distributed along these extremes). Then, we can write an optimization problem with control points $p_i$ as decision variables\footnote{Since the B-spline curve is clamped (see for more details \cite{obstacle2015MED}) it means that the extreme control points are already fixed: $\Upsilon_1(p_0=\mathbf z(t_0))=w_0$ and $\Upsilon_1(p_n=\mathbf z(t_N))=w_N$.} whose goal is to minimize a state and/or input integral cost $\Xi(\bold{x}(t),\bold{u}(t))$ along the time interval $[t_0,t_N]$:
\begin{equation}
\label{eq:flatcost}
\begin{split}
\mathbf P=&\arg\min\limits_{\mathbf P} \int_{t_0}^{t_N}||\tilde \Xi(\mathbf B_d(t),\mathbf P)||_Qdt,\\
&\textrm{s.t. constraints \eqref{eq:flatconstraints} are verified},
\end{split}
\end{equation}
with $Q$ a positive symmetric matrix. The cost in \eqref{eq:flatcost} can impose any penalization we deem necessary (length of the trajectory, input variation/magnitude, energy minimization and the like). 
\subsection{Flatness-based system representation}
\label{subsec:flatReprez}
By replacing the rotation matrix \eqref{eq:roma} in the translation equation $\eqref{eq:tNewton}$ and disregarding the perturbation force as well as replacing the inertia tensor \eqref{eq:tensor} in the rotation equation \eqref{eq:rNewton}, we obtain the matrix form of the quadcopter dynamics:
\begin{align}
\label{eq:tmatrix}
\begin{bmatrix} \ddot{x} \\ \ddot{y} \\ \ddot{z} \end{bmatrix}&= \begin{bmatrix} 0\\0\\-g \end{bmatrix} + \frac{1}{m}\begin{bmatrix} c\phi s\theta c\psi + s\phi s\psi \\ c\phi s\theta s\psi -s\phi c\psi\\ c\phi c\theta \end{bmatrix}T, \\
\label{eq:rmatrix}
\begin{bmatrix} \dot{\omega_x}\\ \dot{\omega_y}\\ \dot{\omega_z} \end{bmatrix}&=\begin{bmatrix} (I_{yy}-I_{zz})I_{xx}^{-1}\omega_y \omega_z \\ (I_{zz}-I_{xx})I_{yy}^{-1}\omega_z \omega_x\\ (I_{xx}-I_{yy})I_{zz}^{-1}\omega_x \omega_y \end{bmatrix} + \begin{bmatrix} I_{xx}^{-1} \tau_\phi \\I_{yy}^{-1}\tau_\theta\\ I_{zz}^{-1} \tau_\psi\end{bmatrix}.
\end{align}
Considering the nonlinear dynamics \eqref{eq:tmatrix} we derive the following flat output vector $\mathbf z \in \mathbb{R}^{4}$ whose dimension equals to the number of inputs $\begin{bmatrix} T \ \tau_\phi \ \tau_\theta \ \tau_\psi \end{bmatrix}^\top $:
\begin{equation}
\label{eq:flatOuptus}
\mathbf z=\begin{bmatrix} z_1 \ z_2 \ z_3 \ z_4 \end{bmatrix}^\top =\begin{bmatrix} x \ y \ z \ \tan\left(\frac{\psi}{2}\right) \end{bmatrix}^\top,
\end{equation}
which will be used to describe the remaining states and inputs (roll, pitch, yaw, thrust and the like):
\begin{align}
\label{eq:phiflat}
&\phi=\arcsin\left(\frac{2z_4\ddot{z_1}-(1-z_4^2)\ddot{z_2}}{(1+z_4^2)\sqrt{\ddot{z_1}^2+\ddot{z_2}^2+(\ddot{z_3}+g)^2}}\right), \\
\label{eq:thetaflat}
&\theta=\arctan \left( \frac{(1-z_4^2)\ddot{z_1}+2z_4\ddot{z_2}}{(1+z_4^2)(\ddot{z_3}+g)} \right), \\
\label{eq:psiflat}
&\psi=2 \arctan (z_4),\\
\label{eq:thrustflat}
&T=m\sqrt{\ddot{z_1}^2+\ddot{z_2}^2+(\ddot{z_3}+g)^2}.
\end{align}
Gathering the angular velocity detailed in \eqref{eq:omega} into the rotation equation \eqref{eq:rNewton}, we obtain the torques described in term of the angular positions:
\begin{equation}
\label{eq:torqueFlat}
\tau_\eta = ^BI \left( W \ddot{\eta} + \dot{W} \dot{\eta} \right)+(W \dot{\eta}) \times (^BI W \dot{\eta}),
\end{equation}
which can be easily interpreted in the flat output space by introducing \eqref{eq:phiflat}--\eqref{eq:psiflat} and which we do not show here due to their convoluted representation.\\
With respect to the notation in \eqref{eq:diff_a}, mapping $\Upsilon_1(\cdot)$ comes from \eqref{eq:flatOuptus}--\eqref{eq:psiflat} (with a derivation degree $q=3$) and mapping $\Upsilon_2(\cdot)$ from \eqref{eq:thrustflat} and the expansion of \eqref{eq:torqueFlat} (with a derivation degree $q+1=4$). For further use we denote $\Upsilon_\xi(\cdot), \Upsilon_\eta(\cdot), \Upsilon_T(\cdot),\Upsilon_{\tau_\eta}(\cdot)$ the mappings which map $\mathbf z$ into the corresponding variable (e.g., $\xi=\Upsilon_\xi(\mathbf z)$).\\

\begin{remark} 
\label{re:2}
Note that, there exist necessary and sufficient conditions for differential flatness as well as the ''sequential"\footnote{There is no gurantee that this procedure finishes in a finite number of steps \cite{levine2011necessary}.} procedure used to test if the system is flat \cite{levine2011necessary}. Some insights on the procedure are summarized here.

Considering the general system \eqref{eq:QuadG} with $f$ smooth, under several specific conditions described in \cite{levine2011necessary}, there exists an \emph{underdetermined} implicit system $F$ with dimension of $n-m$ satisfying \cite{levine2011necessary}:
\begin{equation}
\label{ImpSys}
rank \left(  \frac{\partial f}{\partial 
\bold{u}}  \right)=m  \Leftrightarrow \exists F(\bold{x},\dot{\bold{x}})=0, \  rank \left(\frac{\partial F}{\partial \dot{\bold{x}}}\right)=n-m,
\end{equation}
where $n$ is the number of states and $m$ is the number of inputs. 


Equation \eqref{ImpSys} shows that n-m implicit functions F suffice to express the dynamics of f. Consequently, we may find $m$ variables which can be used to express all the remaining $n-m$ variables. These $m$ variables can be taken as the flat outputs used to describe the rest of the states and the inputs. Similarly, we can follow the sequential procedures provided in \cite{levine2011necessary} to feasibly obtain the flat output representation.

For our particular case \eqref{eq:tmatrix},\eqref{eq:rmatrix}, the two important implicit functions are:
\begin{subequations}
\begin{align}
&s\phi \sqrt{\ddot{x}^2+\ddot{y}^2+(\ddot{z}+g)^2}-s\psi \ddot{x} + c\psi \ddot{y} = 0, \\
&t\theta (\ddot{z}+g)-c\psi\ddot{x}-s\psi\ddot{y}=0.
\end{align}
\end{subequations}
One can easily describe $\phi,\theta$ in terms of the four other states. As a result, the conventional flat output is proposed as $\mathbf z=\begin{bmatrix} z_1 \ z_2 \ z_3 \ z_4 \end{bmatrix}^\top =\begin{bmatrix} x \ y \ z \ \psi \end{bmatrix}^\top$ which has been researched before \cite{bipin2014autonomous,rivera2010flatness,formentin2011flatness}. We found that the `naive' approach of taking $z_4=\psi$ leads to extremely convoluted calculations, therefore, we introduce a new formulation in \eqref{eq:flatOuptus}. 
\ensuremath{\hfill\square}
\end{remark}
\begin{remark}
Other remarks can be made over the simplified version of flat representation usually employed in the state of the art \cite{bipin2014autonomous,rivera2010flatness,formentin2011flatness}. Assuming that yaw angle equals to zero, the formulation \eqref{eq:phiflat},\eqref{eq:thetaflat} simplifies to: 
\begin{subequations}
\begin{align}
\phi &= \arcsin\left(\frac{-\ddot{z_2}}{\sqrt{\ddot{z_1}^2+\ddot{z_2}^2+(\ddot{z_3}+g)^2}}\right),\\
\theta &=\arctan \left( \frac{\ddot{z_1}}{\ddot{z_3}+g} \right).
\end{align}
\end{subequations}
The problem is that while tracking this trajectory the real dynamics will actually vary the yaw angle (a possible solution not followed here is to track $\psi=0$ at the runtime). We do not follow these assumptions in the present paper since we want to exploit all the degrees of freedom, thus fully taking into account the nonlinear system dynamics (including the yaw angle).\ensuremath{\hfill\square}
\end{remark}
Solving \eqref{eq:flatcost} over a B-spline parametrization as in Section \ref{subsec:def} with the flat representation from \eqref{eq:flatOuptus}--\eqref{eq:torqueFlat} we have the general mapping:
\begin{subequations}
\label{eq:flatvariable}
\begin{align}
&\bar \xi =\Upsilon_\xi(\bar{\mathbf z}),\\
&\bar \eta = \Upsilon_\eta(\bar{\mathbf z}),\\
&\bar T=\Upsilon_T(\bar{\mathbf z}),\\
&\bar \tau=\Upsilon_\tau(\bar{\mathbf z}),
\end{align}
\end{subequations}
where the flat output $\bar{\mathbf z}$, the flat states $\bar \xi,\bar \eta$ and the flat inputs $\bar T,\bar \tau$ are given by \eqref{eq:flatOuptus}--\eqref{eq:torqueFlat}.

\section{Control design for trajectory tracking}
\label{sec:control}
This section introduces first the general control strategy usually employed in the literature for a quadcopter system. Next, we propose two control design strategies based on the concept of feedback linearization and facilitated by the flatness construction detailed in Section \ref{subsec:flatReprez}. These first two strategies built for two different missions, control the attitude and the torques of the quadcopter, pave the way for additional control strategies which make more use of the information provided by the a priori generated flat trajectory, i.e., positions, angles, accelerations, thrust force. The idea behind the next three strategies is to use the attitude and torque controllers combined with appropriate input references obtained from the flatness procedure introduced in Section \ref{sec:flat}. 
\subsection{General control scheme}
\label{subsec:classical}
A typical control scheme for quadcopters (and UAV systems in general) is depicted in Figure \ref{fig:control}.
\begin{figure}[ht!]
\begin{center}
\includegraphics[width=\columnwidth]{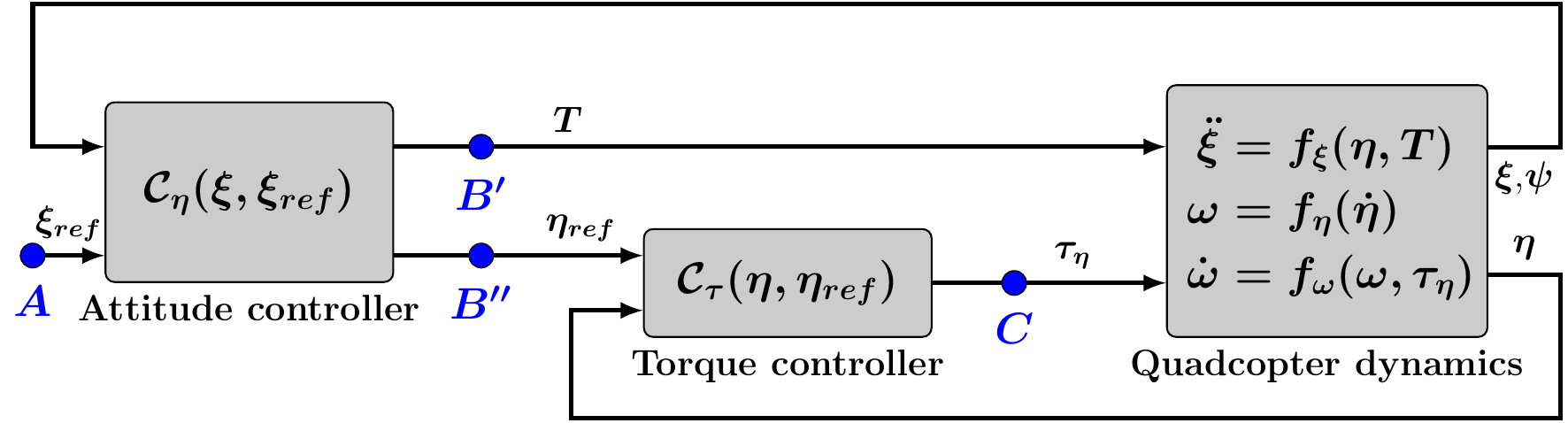}
\caption{Control scheme for a quadcopter system.}
\label{fig:control}
\end{center}
\end{figure}
The preferred approach is to consider two control layers, thus exploiting the decoupling between the translational and rotational dynamics of the quadcopter. At the higher level, an \emph{attitude controller} $\mathcal C_\eta(\xi,\xi_{ref})$ compares an externally given reference position $\xi_{ref}$ with the real position $\xi$ and provides outputs reference angles $\eta_{ref}$ and thrust $T$. The latter is sent directly to the quadcopter and the former to the lower level \emph{torque controller} $\mathcal C_\tau(\eta,\eta_{ref})$ which compares it with the real angles $\eta$ in order to provide the angle torques $\tau_\eta$. 

As also underlined in the control schema of Figure \ref{fig:control}, the attitude controller provides attitude and
thrust references. Usually this allows for simple movements like straight line tracking, circular movement around a fixed center, hovering at a fixed height and so forth. The torque controller provides the torques $\tau_\eta$ which enforce the quadcopter angular positions $\eta$ to track their references $\eta_{ref}$. Standard control design methods for these two controllers can be found in \cite{formentin2011flatness,sydney2013dynamic}. In these studies, they used the simplified model with zero yaw angle to obtain the attitude controller and classical PID regulator inside the torque controller. As mentioned before in Section \ref{sec:model}, the quadcopter rotation dynamics are nonlinear system and are not suitable for linear PID controller. In what follows, we provide effective constructions for the torque and attitude controllers based on feedback linearization which take into account the quadcopter dynamics. These will be introduced in the forthcomming control design strategies via flatness for trajectory tracking. 

\subsection{Torque controller}
\label{subsec:RotC}
The proposed torque controller design builts upon the computed torque control concept which is a special application of feedback linearization of nonlinear systems (basic notions and details on feedback linearization and, in particular, computed torque control can be found in \cite{craig2005introduction}, \cite{tzafestas2013introduction}). It has gained popularity in modern system theory by providing excellent tracking performance through nonlinear compensations (assuming a precise dynamic model is available \cite{jang2015computed}). 

Consider the reformulation of the rotational dynamics \eqref{eq:torqueFlat} as:
\begin{equation}
\label{eq:tor1}
M(\eta)\ddot \eta+V(\eta,\dot \eta)=\tau_\eta,
\end{equation}
with mappings $M(\eta)$, $V(\eta,\dot \eta)$ of appropriate content, i.e, $M(\eta)=^BIW$ and $V(\eta,\dot \eta)=^BI \dot{W} \dot{\eta} +(W \dot{\eta}) \times (^BI W \dot{\eta})$.

By using the partitioned controller scheme introduced in \cite{craig2005introduction}, we take the control law for angle torques as:
\begin{equation}
\label{CTClaw}
\tau_\eta= \alpha \tau' +\beta,
\end{equation}
where $\alpha,\beta$ named model-based portion and $\tau'$ named servo portion are taken as:
\begin{subequations}
\begin{align}
&\alpha=M(\eta),\\
&\beta=V(\eta,\dot \eta),\\
&\tau'=\ddot \eta_{ref}+K_d\dot \epsilon_\eta+K_p\epsilon_\eta+K_i\int \epsilon_\eta dt,
\end{align}
\end{subequations}
with $\epsilon_\eta =\eta_{ref}-\eta$. Introducing \eqref{CTClaw} into \eqref{eq:tor1} leads to a linear error dynamics:
\begin{equation}
\label{eq:linearizederror}
\ddot \epsilon_\eta+K_{d\eta}\dot \epsilon_\eta+K_{p\eta}\epsilon_\eta+K_{i\eta}\int \epsilon_\eta dt=0.
\end{equation}
Note that suitable parameters $K_{p\eta}$, $K_{d\eta}$, $K_{i\eta}$ (diagonal matrices from $\mathbb{R}^{3}$) need to be chosen in \eqref{eq:linearizederror} so that the system is stable. To this end the following proposition is introduced.
\begin{proposition}
\label{prooferror}
Consider a third order linear dynamic system with the bounded and continuous input $U$ (e.g., perturbation triggered by a bounded and continuous wind gust) and the output $E$ which is the scalar error between specific state and its reference:
\begin{equation}
\label{errordynamic}
\ddot{E}+K_d \dot{E}+K_p E+K_i \int E dt=U.
\end{equation}
By choosing the scalar parameters $K_p$, $K_d$, $K_i$ satisfying the conditions:
\begin{equation}
\label{condition}
\begin{cases}
     K_p, K_d, K_i >0   \\
     K_pK_d>K_i\\
\end{cases},
\end{equation}
the system \eqref{errordynamic} is uniformly asymptotically stable.
\end{proposition}

\begin{proof} Gathering $k(t)=\int_{0}^{t} E(\tau) d\tau \Leftrightarrow E=\dot{k}$ into \eqref{errordynamic}, we arrive to the new system in terms of $k(t)$:
\begin{equation}
\label{inferrordynamic}
k^{(3)}+K_d \ddot{k}+K_p \dot{k}+K_i k=U.
\end{equation}
Applying the Laplace transform of $K(s)=\mathcal{L}(k(t))$ and $U(s)=\mathcal{L}(U(t))$ for \eqref{inferrordynamic}, we get:
\begin{align}
\nonumber
&s^3K(s)+K_d s^2K(s)+K_p sK(s)+K_iK(s)=U(s) \\
\label{Sequation}
\Rightarrow& \frac{K(s)}{U(s)}=\frac{1}{s^3+K_d s^2+K_p s+K_i}.
\end{align}
This linear time-invariant system \eqref{inferrordynamic} is BIBO stable , or in other words, the characteristic equation has all its roots with negative real parts if and only if  parameters $K_p,K_d,K_i$ satisfying condition \eqref{condition} which is the Routh--Hurwitz criterion.

With bounded input $U$, the system results in bounded output $k(t)$ over the time interval $[t_0,\infty)$:
\begin{align}
&||k(t)||\leq C\ \forall t \in [t_0,\infty), \ C \in \mathcal{R} \\
\label{inferrorbound}
\Rightarrow&||\int_{0}^{t} E(\tau) d\tau ||\leq C\ \forall t \in [t_0,\infty), \ C \in \mathcal{R}.
\end{align}
Next, we use the Barbalat's lemma \cite{tzafestas2013introduction} which states that a continuous function $f(t)$ satisfying $\lim_{t\to\infty} f(t) = \alpha, \ \alpha < \infty$, its continuous derivative $f'(t)$ satisfies $\lim_{t\to\infty} f'(t) = 0$. Consequently, mapping $f(t),f'(t)$ to appropriate contents, e.g., $\int_{0}^{t} E(\tau) d\tau$ and $E(t)$ respectively, we already obtained the condition \ref{inferrorbound} and since $U(t)$ is continuous, it leads to the continuous $E(t)$. As the result, we come to the conclusion $\lim_{t\to\infty} E(t) = 0$. Thus completing the proof.
\end{proof}
\subsection{Attitude controller}
\label{subsec:attitudeC}
In general, the attitude controller provides the thrust force $T$ and the angle references $\eta_{ref}$ which are necessary for the quadcopter to follow the position references $\xi_{ref}$. The proposed attitude controller design is also based on the concept of feedback linearization of nonlinear systems which will drive the translation dynamics to error dynamics similar with those in \eqref{eq:linearizederror}. 

Considering the roll, pitch, yaw angles and input thrust $T$ in terms of the flat output described in equations \eqref{eq:phiflat}--\eqref{eq:thrustflat}, they can be particularly expressed as $\phi=\Gamma_\phi(\ddot{z_1},\ddot{z_2},\ddot{z_3},z_4)$, $\theta=\Gamma_\theta(\ddot{z_1},\ddot{z_2},\ddot{z_3},z_4)$, $\psi=\Upsilon_\psi(z_4)$ and $T=\Gamma_T(\ddot{z_1},\ddot{z_2},\ddot{z_3})$. We provide the reference to be followed (the output of the attitude controller from the scheme in Figure \ref{fig:control}) as:
\begin{subequations}
\label{controllawAC}
\begin{align}
\label{eq:phiCTC}
&\phi_{ref}=\Gamma_\phi(\ddot{z_1}^*,\ddot{z_2}^*,\ddot{z_3}^*,z_4),\\
\label{eq:thetaCTC}
&\theta_{ref}=\Gamma_\theta(\ddot{z_1}^*,\ddot{z_2}^*,\ddot{z_3}^*,z_4),\\
\label{eq:psiCTC}
&\psi_{ref}=\Upsilon_\psi(\bar{z_4}), \\
\label{eq:thrustCTC}
&T=\Gamma_T(\ddot{z_1}^*,\ddot{z_2}^*,\ddot{z_3}^*),
\end{align}
\end{subequations}
where the corrective term $\xi^* \triangleq \begin{bmatrix} z_1^*\ z_2^* \ z_3^* \end{bmatrix}^\top $ is given as:
\begin{equation}
\label{eq:correctiveterm}
\xi^*=\xi_{ref}+K_{d\xi}\int \epsilon_\xi dt+K_{p\xi}\int \int\epsilon_\xi dt+K_{i\xi}\int \int \int \epsilon_\xi dt,
\end{equation}
with $\epsilon_\xi =\xi_{ref}-\xi$ and $K_{p\xi}$, $K_{d\xi}$, $K_{i\xi}$ are diagonal matrices from $\mathbb{R}^{3}$. \\

Taking into account the external perturbation force, gathering equations \eqref{controllawAC} into the translation equation \eqref{eq:tNewton} leads to the following relation:
\begin{equation}
\ddot{\xi}-\frac{\overrightarrow{F_D}}{m}=\ddot{\xi}_{ref}+K_{p\xi}\epsilon_\xi+K_{d\xi}\dot \epsilon_\xi+K_{i\xi}\int \epsilon_\xi dt,
\end{equation}
which results in the error dynamics:
\begin{equation}
\label{eq:Positionerror}
\ddot{\epsilon_\xi}+K_{p\xi}\epsilon_\xi+K_{d\xi}\dot \epsilon_\xi+K_{i\xi}\int \epsilon_\xi dt=-\frac{\overrightarrow{F_D}}{m},
\end{equation}
similarly to Proposition \ref{prooferror}.

In what follows, the attitude and torque controllers will prove useful for additional strategies which allow feedback control via planned flat trajectory. 
\subsection{Flat angle tracking}
\label{subsubsection:CTCFlat}
Starting from the lower level and using only the torque controller introduced in Section \ref{subsec:RotC}, it is possible to control the quadcopter by providing directly the input components $T$ and $\eta_{ref}$ obtained by the flatness-based trajectory generation (insertion at points $B'$ and $B''$ in Figure \ref{fig:control}):
\begin{equation}
\label{eq:controlTorque}
T=\bar{T}, \: \eta_{ref}=\bar{\eta}.
\end{equation}
Then, the torque controller gives the angle torques $\tau_\eta$ as detailed in Section \ref{subsec:RotC}. According to Proposition \ref{prooferror}, the quadcopter rotating system will be asymptotically stable.

Applying this strategy, the angle tracking leads actually to the position tracking in the predicted case \footnote{The predicted case is the combination of reference trajectory coming with wind information used in the flatness procedure.}. Note that from a practical viewpoint, this strategy is is realistic for small-scale quadcopters, (e.g., flycam, radio controller quadcopter) since the angle feedback can be approximately obtained by available sensors such as gyroscope, accelerometer and geomagnetic field sensor, while the position feedback is difficult to retrieve. It is worth underlining that this open-loop functioning for position is sensitive to disturbances and other sources of error. To counteract this limitation, in the next section we will introduce a position feedback loop.
\subsection{Flat position tracking}
\label{subsubsection:CTCPosition}
This controller which is based on the attitude controller presented in Section \ref{subsec:attitudeC}, compares the reference $\bar{\xi}$ and real position $\xi$ and provides the thrust force $T$ and angle torques $\tau_\eta$. The general idea is well illustrated in Figure \ref{fig:control}:
\begin{itemize}
\item[-] The trajectory generation provides the references $\bar{\xi}$ and $\bar{z_4}$ (insertion at point A) as in \eqref{eq:flatvariable}.
\item[-] The attitude controller provides thrust force (insertion at point B') and necessary angles $\eta_{ref} \triangleq \begin{bmatrix} \phi_{ref} \ \theta_{ref} \ \psi_{ref} \end{bmatrix}^\top$ as introduced in equation \eqref{controllawAC} but in terms of $\bar{z_4}$ in stead of $z_4$.
\item[-] The angle torques $\tau_\eta$ are calculated based on the rotation equation \eqref{eq:torqueFlat} in terms of $\eta_{ref}$, then, sent to the quadcopter system (insertion at point C).
\end{itemize}

This controller, as we will also validate through simulations, achieves the good results for position tracking. Note that, the quadcopter position can be straightforward to be observed by using GPS (Global Positioning System). One solution is through the use of the technique called \emph{differential GPS} or \emph{dual frequency GPS} which gives a resolution of 1 $m$, if a second static receiver at a known exact position is employed \cite{tzafestas2013introduction}. However, the open-loop functioning for angle of this strategy generates various errors of yaw angle $\psi$. 

Next, a combination of the two above procedures is discussed.

\subsection{Combined flat angle and position tracking}
\label{subsubsec:CTCFAT}
Considering the two previous strategies, we recognized the necessity of both position and angle feedback. Hence, this controller design follows the two-layer classical control strategy described in \ref{subsec:classical}. More precisely, at the upper level we use the attitude controller detailed in Section \ref{subsec:attitudeC} which compares the position reference $\bar\xi$ and the real position $\xi$ to provide the thrust $T$ (\ref{eq:thrustCTC}) and the reference angles $\eta_{ref}=\begin{bmatrix} \phi_{ref} \ \theta{ref} \ \psi_{ref} \end{bmatrix}^\top$ (\ref{eq:phiCTC}--\ref{eq:psiCTC}). The angles are sent to the lower level which is the torque controller detailed in Section \ref{subsec:RotC}. The torque controller provides the angle torques $\tau_\eta$ to the quadcopter system. Note that, the quadcopter position feedback is necessary for the attitude controller and the orientation for the torque controller. Assuming we have at our disposal all of the ideal necessary sensors, this strategy provides the best trajectory tracking results which will be demonstrated and compared in the next section.

\section{Simulation results and comparison}
\label{sec:sim}
In this section, we first present simulation results of our  control strategies introduced in section \ref{sec:control}. Then, various comparisons of our contributions with other flatness-based control approaches \cite{chamseddine2012flatness,rivera2010flatness,formentin2011flatness,mellinger2011minimum}
 are provided.
\subsection{Simulation results}
This section presents extensive simulation results for a Crazyflie $2.0$ quadcopter \cite{Quad} characterized by the following parameters:
\begin{itemize}
\item[-] each of the four propellers has bounds on the (load-free) rotating speed $|\omega_i| \leq 58800$ $[rpm]$ and angular acceleration $|\Delta \omega_i|\leq 1000$ $[rad/s^2]$;
\item[-] $g=9.81$ $[m/s^2]$, $m=0.5$ $[kg]$, $L=0.225$ $[m]$, $K_T=2.98\times 10 ^{-6}$ $[kgm]$, $b=1.14\times10^{-7}$ $[kgm^2]$, $I_{xx}=I_{yy}=4.856\times10^{-3}$ $[kgm^2]$, $I_{zz}=8.801\times 10^{-3}$ $[kgm^2]$.
\end{itemize}
The simulation scenarios consider a collection of way-points $\mathbb W=\{\begin{bmatrix} 0 & 0 & 5 \end{bmatrix}^\top, \\ \begin{bmatrix} 0.4 & 0.9 & 6 \end{bmatrix}^\top, \begin{bmatrix} 1.4 & 1.2 & 6.5 \end{bmatrix}^\top, \begin{bmatrix} 2 & 0.8 & 5.7 \end{bmatrix}^\top, \begin{bmatrix} 1.5 & -0.5 & 5 \end{bmatrix}^\top\}$ with the associated time instants $\{0,3,5.5,7,10\}$ second.

We implement the optimization problem \eqref{eq:flatcost} by choosing to minimize the total trajectory length and to pass through the a priori given way-points in a total time $T=10s$. We consider B-spline basis functions of degree $d=6$ and a collection of $12$ control points as in \eqref{eq:control_points} for the flat output parametrization. The resulting trajectory, position, angles\footnote{We used a standard polynomial function $z_4(t)$ for $\psi(t)$ to smoothly increase from 0 to 10 degrees in 10s.}, thrust and torques are depicted in Figure \ref{fig:flat} and Figure \ref{fig:bla}. 
\begin{figure}[h!]
\begin{center}
\includegraphics[width=1\columnwidth]{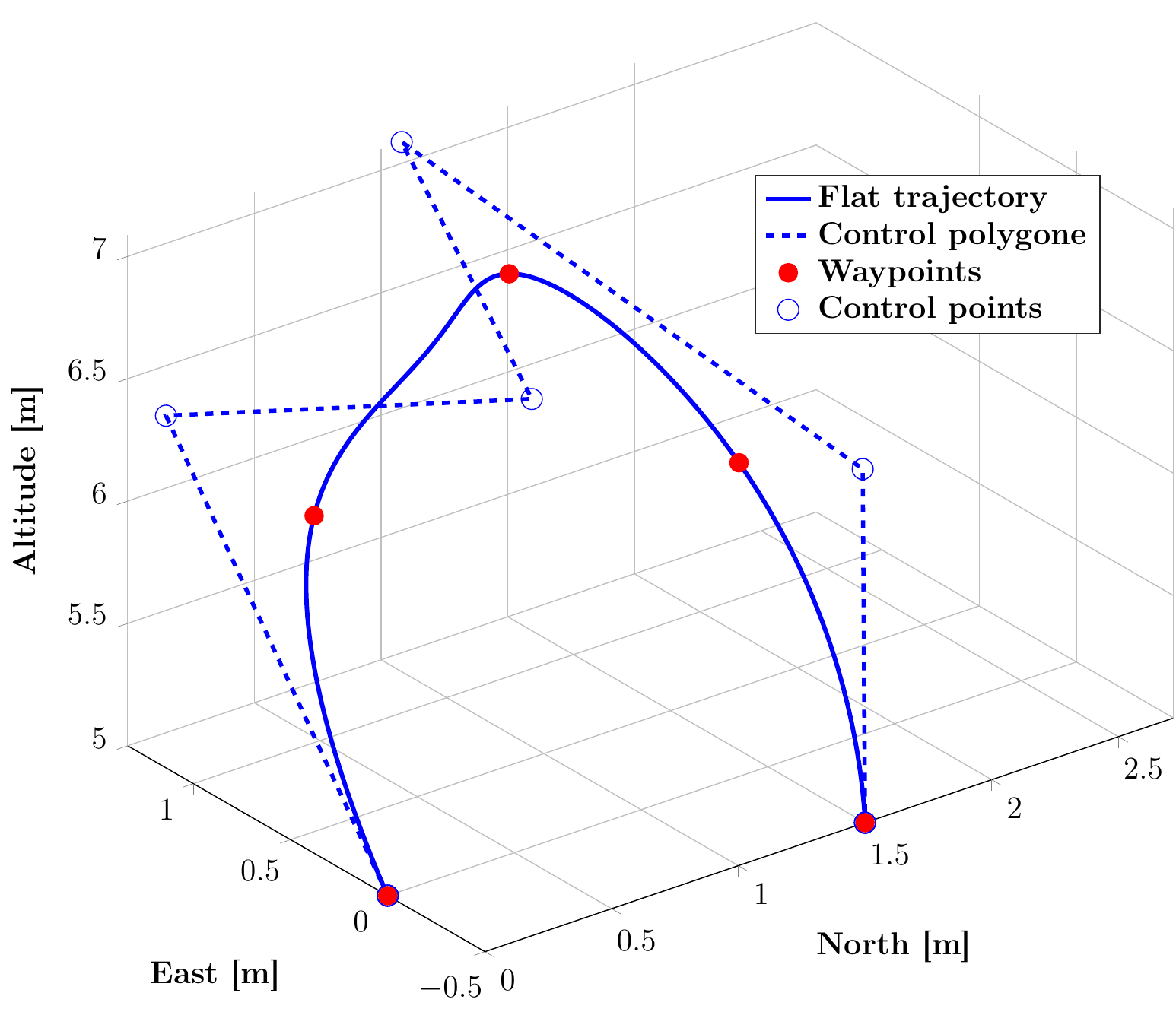}
\caption{B-spline parametrized flat trajectory with the associated control points and passing through way-points.}
\label{fig:flat}
\end{center}
\end{figure}

\begin{figure*}
\begin{center}
\subfloat[Flat position on the three axes and the control points.]{\label{fig:flatPosition}\includegraphics[width=1\textwidth]{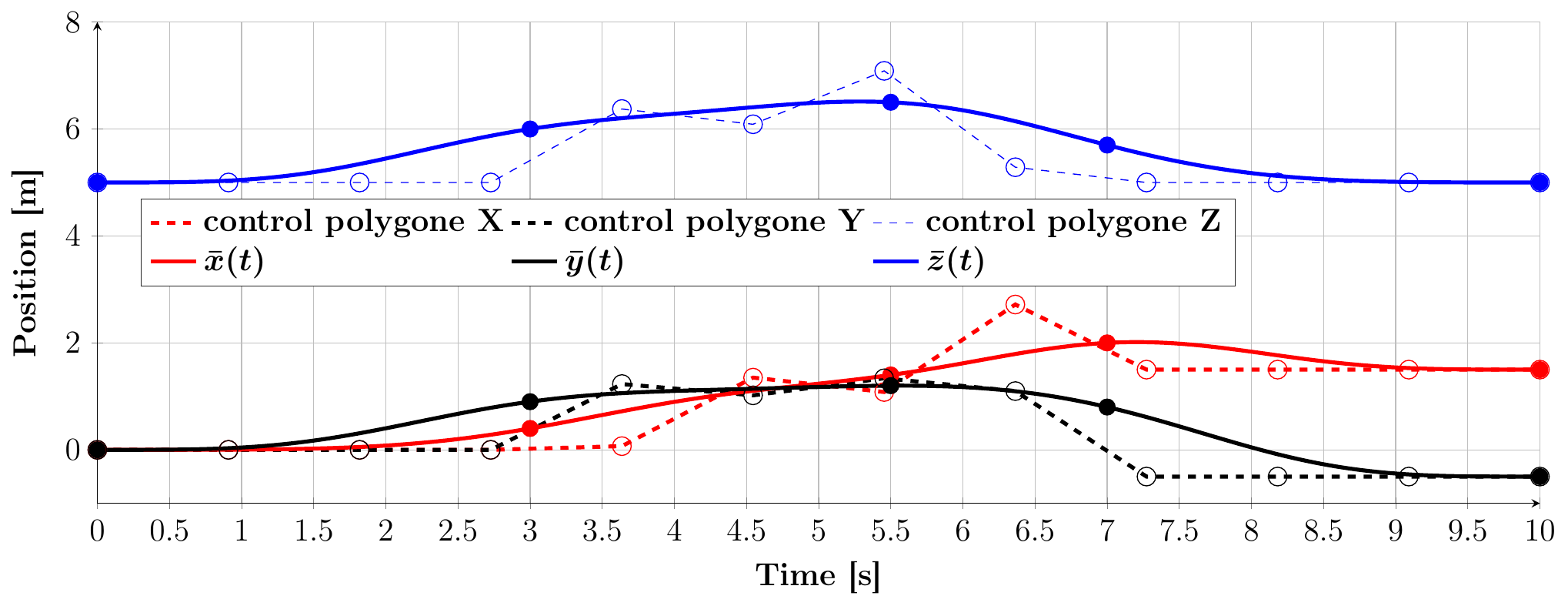}}\\
\subfloat[Flat reference angles, roll, pitch and yaw.]{\label{fig:flatAngles}\includegraphics[width=1\textwidth]{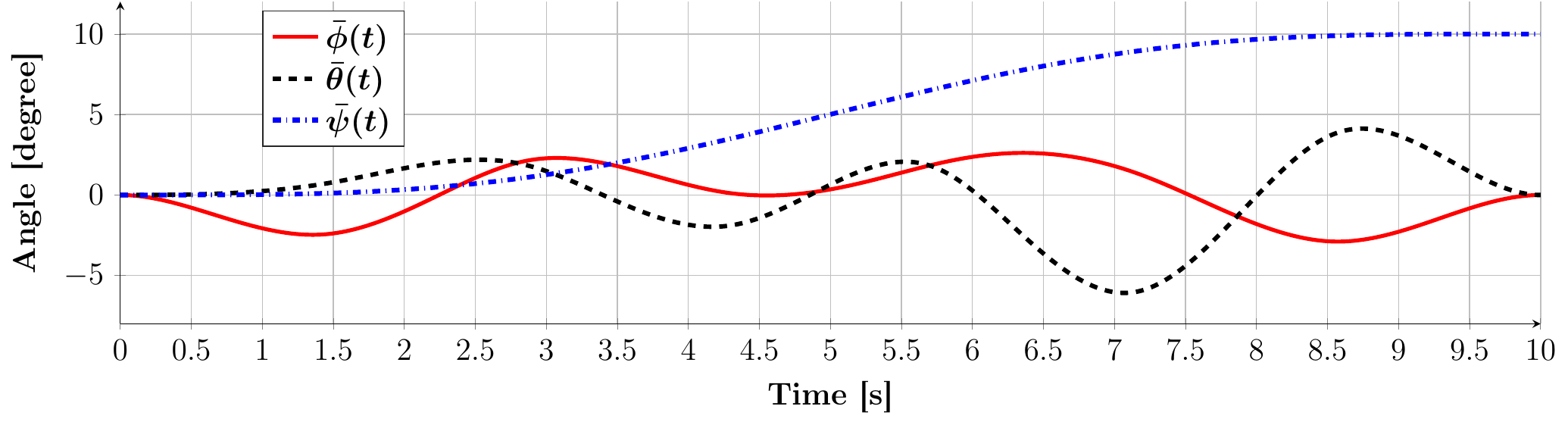}}\\
\subfloat[Flat reference of the thrust force.]{\label{fig:flatThrust}\includegraphics[width=1\textwidth]{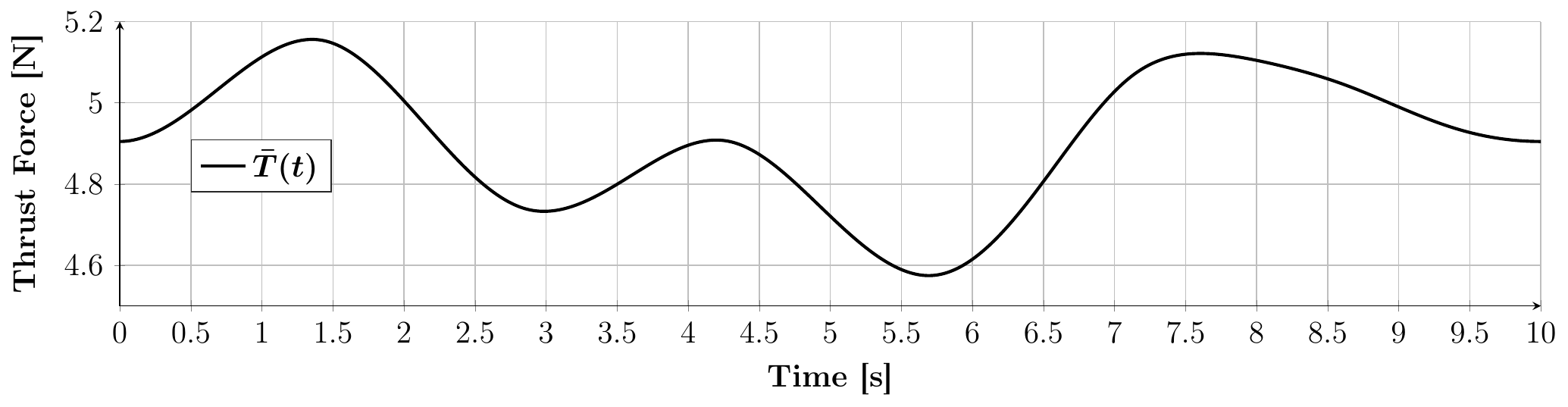}}\\
\subfloat[Flat reference of the angle torques.]{\label{fig:flatTorque}\includegraphics[width=1\textwidth]{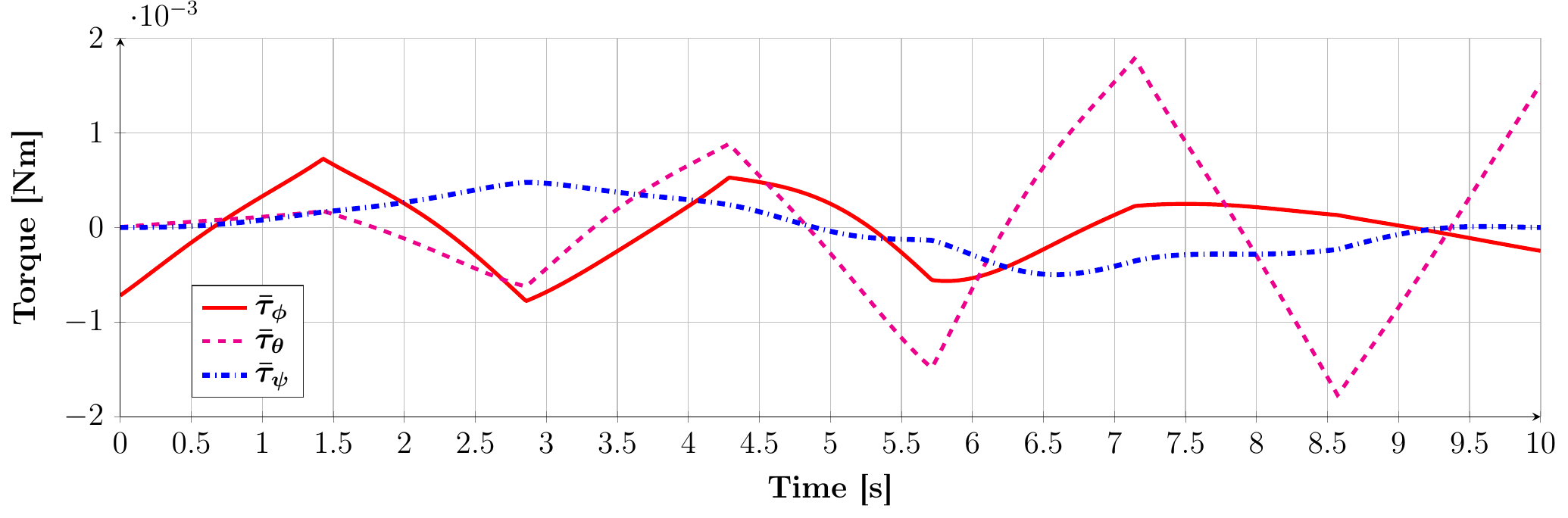}}
\caption{Flat references for positions, angles, thrust and angle torques.}
\label{fig:bla}
\end{center}
\end{figure*}

In what follows we consider the various control strategies discussed in Section \ref{sec:control} and apply them for the reference trajectory generated earlier. For each of these approaches we consider two cases of no wind and wind profile with a maximum speed up to $25$ [$km/h$] (the values are taken from \url{www.meteoblue.com} over the year 2015 in Rh\^one Alpes region, France). The control algorithms implementation are done using Yalmip \cite{YALMIP}, MPT Toolboxes \cite{MPT3} in Matlab/Simulink 2015a over a horizon of $T=10$ sec with a fixed sampled time of $0.01$ sec. The tuning parameters $K_p$, $K_i$, $K_d$ of each controller are delineated in Table \ref{tab:Para}. 
\begin{table}
\centering
\begin{tabular}{|c|c|c|c|}
\hline
\text{Control Scheme}           &$K_p$      &$K_d$		&$K_i$      \\\hline
Torque controller \ref{subsec:RotC}  &   &  &  \\
used in \ref{subsubsection:CTCFlat},\ref{subsubsec:CTCFAT} & $diag\{225,225,225\}$ & $diag \{30,30,30\}$ & $diag\{0,0,0\}$\\\hline
Attitude controller \ref{subsec:attitudeC}     & 	  & 	& \\
used in \ref{subsubsection:CTCPosition},\ref{subsubsec:CTCFAT}	& $diag\{25,25,9\}$  &  $diag\{10,10,6\}$ & $diag\{1,1,0.3\}$ \\  \hline
\end{tabular}
\caption{Parameters of rotation and attitude controller}
\label{tab:Para}
\end{table}
\begin{table}
\centering
\begin{tabular}{|c|c|c|}
\hline
\text{Controller}           &IAE        &  IAE \\
							&no wind    & wind gust \\\hline
Flat angle tracking \ref{subsubsection:CTCFlat} & 0.0151      & 52.2087 \\ \hline
Flat position tracking \ref{subsubsection:CTCPosition} & 0.7210      & 0.9419 \\ \hline
Combined flat angle	and & 0.0227 	  & 0.6221  \\
position tracking  \ref{subsubsec:CTCFAT}  & &     \\\hline 
 
\end{tabular}
\caption{Integral Absolute magnitude of Errors (IAE) of positions \\using the control strategies in Section \ref{sec:control}.}
\label{tab:IAE}
\end{table}
\begin{figure}[h!]
\begin{center}
\includegraphics[width=1\columnwidth]{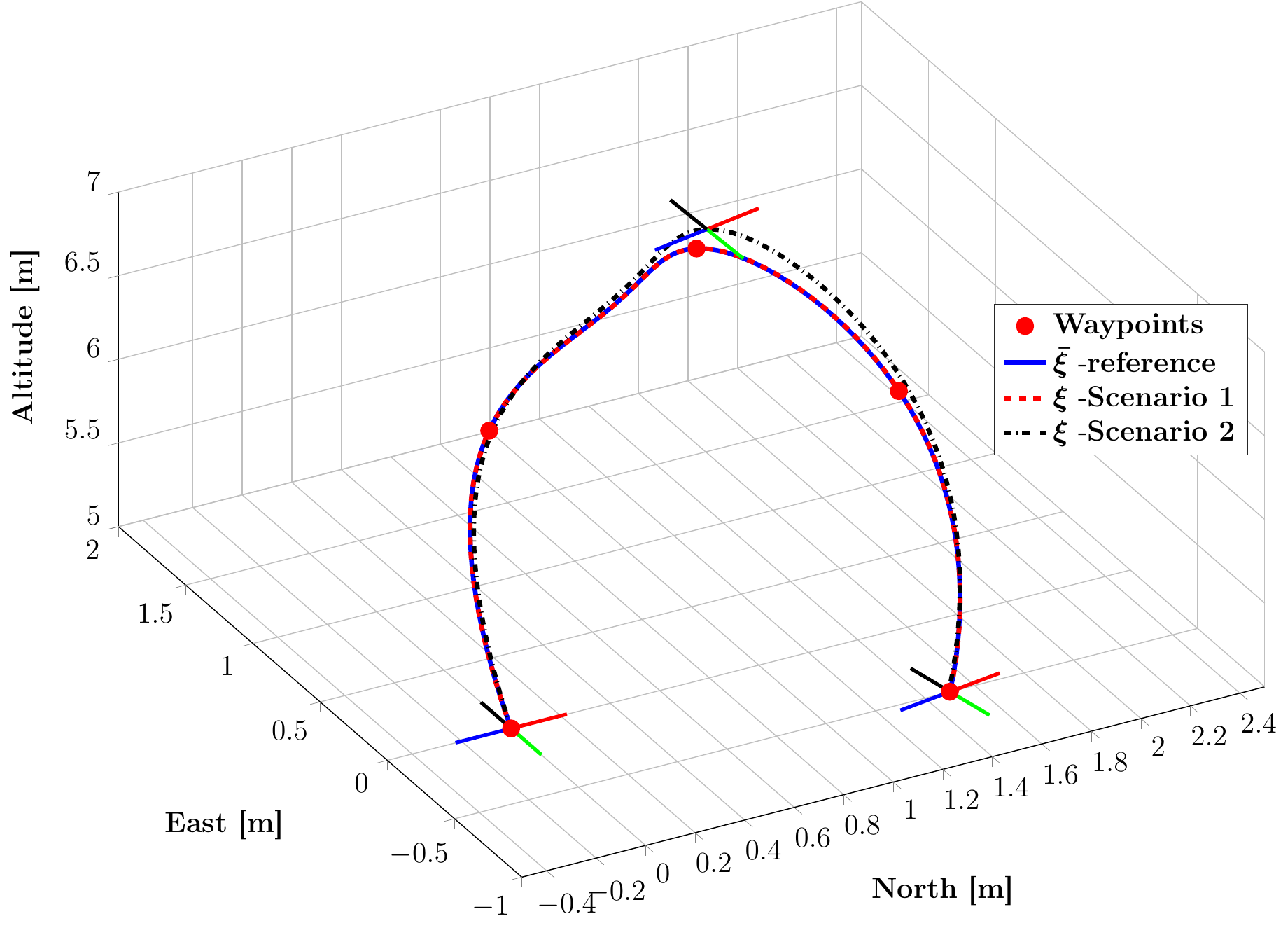}
\caption{Trajectories of quadcopter under different scenarios.}
\label{fig:traj}
\end{center}
\end{figure}

\begin{figure}[h!]
\begin{center}
\includegraphics[width=1\columnwidth]{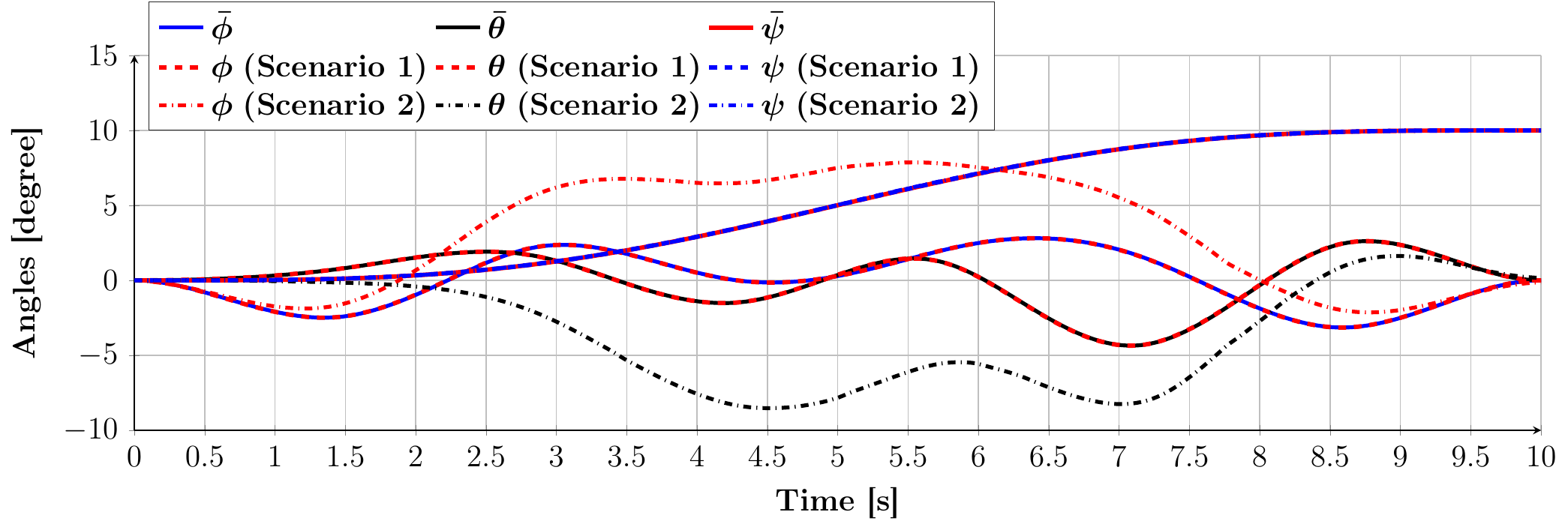}
\caption{Roll, pitch, yaw angles of quadcopter under different scenarios.}
\label{fig:ang}
\end{center}
\end{figure}
For comparison, in each simulation case we take the Integral of Absolute magnitude of the Error (IAE) over the position: $IAE=\int_{t_0=0}^{t_f=10} ||\bar\xi-\xi|| dt$. The results are gathered in Table \ref{tab:IAE}, which leads us to several observations. First of all, under nominal functioning (no wind) the three controllers are comparable, i.e, the IAE values are small and not far away from each other, with controller \ref{subsubsection:CTCFlat} being slightly better. However, in the presence of high disturbances, the controller \ref{subsubsection:CTCFlat} fails as we also indicated in Section \ref{subsubsection:CTCFlat} (IAE$=52.2087$ is too large comparing to the others). In contrast, the flat position tracking controller \ref{subsubsection:CTCPosition} and the combined flat angle and position tracking controller \ref{subsubsec:CTCFAT} have proven their trajectory tracking capabilities even for high wind conditions. Note that, the combined flat angle and position tracking controller \ref{subsubsec:CTCFAT} is behaving best, this being, in our opinion, the most effective control strategy.

We provide illustrations of simulation results for two scenarios: 
\begin{enumerate}
\item[-]Scenario 1: the aim is to track the reference using flat angle tracking controller detailed in Section \ref{subsubsection:CTCFlat} with no wind condition;
\item[-]Scenario 2: the aim is to track the reference using combined flat angle and position tracking controller detailed in Section \ref{subsubsec:CTCFAT} in the wind blow condition with maximum wind speed up to 25 [$km/h$];
\end{enumerate}

Figure \ref{fig:traj} illustrates the quadcopter actual motions  resulted for the two scenarios (for the scenario 1 the actual trajectory is plotted in solid blue line and for the scenario 2 in dash-dotted black line) comparing to the reference trajectory given in dash-dotted red line. It can be seen that the differences w.r.t. the reference are very small although the maximum wind speed of $25$ [$km/h$] is such a difficult condition for trajectory tracking of small-scale UAV in general. 

For the scenario 2, Figure \ref{fig:traj} and \ref{fig:ang} illustrate the quadcopter blown in the positive directions of the IF (East--North--Up coordinate) due to a wind profile from north-east. Figure \ref{fig:ang} proves the effectiveness of the combined flat angle and position tracking controller \ref{subsubsec:CTCFAT} which appropriately tilted the quadcopter to fight against the northeast wind blow.

Our simulations have proven that the various control strategies described in the paper are all capable to track the trajectory in the nominal case and, with specific degree of accuracy, the flat position tracking controller \ref{subsubsection:CTCPosition} and the combined flat angle and position tracking controller \ref{subsubsec:CTCFAT} are effective to track the trajectory in more challenging conditions. The robustness of the controllers can be further enhanced by choosing different corrective terms and/or different parameters.
\subsection{Discussions}
Detailed comparisons are difficult to provide since most of the papers treating this topic provide incomplet data for the flatness generation and inner control loops design makes a point-by-point simulation hard to accomplish, we note, however, several remarks which prove the novel elements of our flatness-based control approach with respect to the references \cite{chamseddine2012flatness,rivera2010flatness,formentin2011flatness,mellinger2011minimum}.
\begin{enumerate}
\item The flat trajectories generated are not always used in simulation. For example, in \cite{formentin2011flatness} the reference tracked is actually a sequence of delayed step functions. In our opinion this actually discards the major advantage of flat parametrizations, that is, of ensuring a feasible trajectory. 
\item The flat output parametrizations often use simple representations (monomials in \cite{chamseddine2012flatness} or cubic splines in \cite{rivera2010flatness}). These implementations strongly limit the number of constraints which can be considered and may lead to numerical issues. In contrast, the b-spline parametrization used in this paper offers smoothness guarantees, is impervious to the number of constraints (in the sense that the degree of the functions does not depend on them) and, most importantly, offers an analytical framework for cost minimizations (e.g., for trajectory length).
\item All flat implementations encountered in the literature consider symplifying assumptions (yaw angle kept constant, small angles, etc). In contrast, the flat representation proposed in \ref{eq:flatOuptus} and presented in Section \ref{sec:flat} can provide explicit (and free of trigonometric terms) descriptions of all state and input components of the quadcopter dynamics. In particular, the angles and torques have a more compact representation, see \ref{appendix} for details. While the resulting flat representations (especially for torques) are still cumbersome, they are nonetheless much more compact than the representations which assume the standard flat output detailed in Remark \ref{re:2} (quantitatively, the difference in formulation length is of an order of magnitude). Not in the least, the novel approach proposed here can be easily employed in similar schemes proposed in the literature \cite{chamseddine2012flatness,rivera2010flatness,formentin2011flatness,mellinger2011minimum} and will lead to simpler formulations and thus to more efficient control loops.
\item In many cases (and in the strategy proposed in Section \ref{sec:control}) the flat angles are used as references for low-level control loops (the ones providing the angle torque values). Many papers employ PD or similar control schemes \cite{formentin2011flatness,mellinger2011minimum} which we consider to be an inferior alternative to the torque control approach \ref{subsec:RotC} proposed here. Assuming accurate parameter measurements this strategy provides a closed-loop linearized rotation dynamic which can handle abrupt reference changes and has a good tracking performances.
\item An aspect sometimes neglected \cite{chamseddine2012flatness,rivera2010flatness} is the difference between angular velocities and the Euler angles rates (an acceptable assumption for small roll and pitch values). While this simplification leads to simpler torque and angle formulations it becomes imprecise at large roll and pitch values and leads thus to imprecise angle tracking. Therefore, while \cite{chamseddine2012flatness,rivera2010flatness} propose strategies similar with the flat position tracking from Section \ref{subsubsection:CTCPosition}, our approach can accurately handle the nonliniearities introduced by the Euler angles and permits to track the position components.
\end{enumerate}
\section{Conclusions}
\label{sec:concl}
This paper addressed the challenging trajectory tracking problem for quadcopter systems using an effective combination between differential flatness and feedback linearization. Classified as severely underactuated systems, detailed kinematic and dynamical models of a quadcopter were required. Next, a reference trajectory was generated off-line using an original flat representation. On-line, feedback linearization-based controllers via flatness were designed for tracking the feasible reference. The power of flat output characterization allowing full flat parametrization of states and inputs, and the state feedback control methods applied for the original nonlinear quadcopter dynamics system without any loss of precision shows promise. These were detailed and validated through proof of concept examples, illustrations and simulation results. 

The original contributions stem from:
\begin{itemize}
\item the flat trajectory construction for the strongly nonlinear quadcopter system which provided positions, angles, thrust and torques;
\item the control strategies based on feedback linearization (i.e., flat angle tracking, flat position tracking) which allowed both orientation and position control without any simplification on the quadcopter system.
\end{itemize}
Future work will concentrate on the introduction of bounded/stochastic disturbances and trajectory reconfiguration mechanisms.
\section*{References}

\bibliography{ref}
\newpage
\appendix
\section{Flat representation \eqref{eq:phiflat}--\eqref{eq:torqueFlat} for the quadcopter dynamics \eqref{eq:tmatrix}--\eqref{eq:rmatrix} }
\label{appendix}
\subsection{Position, angle and thurst components of the quadcopter dynamics \eqref{eq:tmatrix}--\eqref{eq:rmatrix}}

Position components expressed in term of the flat output:
\begin{subequations}
\begin{align}
&x=z_1,\\
&y=z_2,\\
&z=z_3.
\end{align}
\end{subequations}

Angle components expressed in term of the flat output:
\begin{subequations}
\begin{align}
&\phi=\arcsin\left(\frac{2z_4\ddot{z_1}-(1-z_4^2)\ddot{z_2}}{(1+z_4^2)\sqrt{\ddot{z_1}^2+\ddot{z_2}^2+(\ddot{z_3}+g)^2}}\right), \\
&\theta=\arctan \left( \frac{(1-z_4^2)\ddot{z_1}+2z_4\ddot{z_2}}{(1+z_4^2)(\ddot{z_3}+g)} \right), \\
&\psi=2 \arctan (z_4).
\end{align}
\end{subequations}

Thurst expressed in term of the flat output:
\begin{equation}
T=m\sqrt{\ddot{z_1}^2+\ddot{z_2}^2+(\ddot{z_3}+g)^2}.
\end{equation}
\subsection{Torques components of the quadcopter dynamics \eqref{eq:tmatrix}--\eqref{eq:rmatrix}}
Torques expressed in term of $[k_1,k_2,k_3,z_4]=[\ddot{z_1},\ddot{z_2},\ddot{z_3},z_4]$:
\begin{equation}
\nonumber
\begin{split}
\normalsize
\tau_\phi&=I_{xx} \Biggl( \dfrac{1}{\sqrt{1-\dfrac{(2z_4k_1-(1-z_4^2)k_2)^2}{(1+z_4^2)^2(k_1^2+k_2^2+k_3^2)}}} \Bigl( \dfrac{1}{(1+z_4^2)\sqrt{k_1^2+k_2^2+k_3^2}}\Bigl(2\ddot{z_4}k_1\\
&+4\dot{z_4}\dot{k_1}+2z_4\ddot{k_1}+2\ddot{z_4}^2k_2+2z_4\ddot{z_4}k_2+4z_4\dot{z_4}\dot{k_2}-(1-z_4^2)\ddot{k_2} \Bigl)\\
&-\dfrac{4 \Bigl(2\ddot{z_4}k_1+2z_4\dot{k_1}+2z_4\dot{z_4}k_2-(1-z_4^2)\dot{k_2} \Bigl)z_4\dot{z_4}}{(1+z_4^2)^2\sqrt{k_1^2+k_2^2+k_3^2}} +\dfrac{8 \Bigl( 2z_4k_1-(1-z_4^2)k_2 \Bigl) z_4^2 \dot{z_4}^2}{(1+z_4^2)^3\sqrt{k_1^2+k_2^2+k_3^2}}\\
&-\dfrac{2\Bigl(2z_4k_1-(1-z_4^2)k_2\Big)\dot{z_4}^2+2\Bigl( 2z_4k_1-(1-z_4^2)k_2\Bigl) z_4 \ddot{z_4}}{(1+z_4^2)^2\sqrt{k_1^2+k_2^2+k_3^2}}\\
&-\dfrac{\Bigl( 2\dot{z_4}k_1+2z_4\dot{k_1}+2z_4\dot{z_4}k_2-(1-z_4^2)\dot{k_2}\Bigl) \Bigl(2k_1\dot{k_1}+2k_2\dot{k_2}+2k_3\dot{k_3} \Bigl)}{(1+z_4^2)(k_1^2+k_2^2+k_3^2)^{3/2}} \\
&-\dfrac{\Bigl(2z_4k_1-(1-z_4^2)k_2\Big)\Bigl(\dot{k_1}^2+k_1\ddot{k_1}+\dot{k_2}^2+k_2\ddot{k_2}+\dot{k_3}^2+k_3\ddot{k_3} \Bigl)}{(1+z_4^2)(k_1^2+k_2^2+k_3^2)^{3/2}} \\
&+\dfrac{2\Bigl( 2z_4k_1-(1-z_4^2)k_2)\Bigl) \Bigl( 2k_1\dot{k_1}+2k_2\dot{k_2}+2k_3\dot{k_3} \Bigl)z_4\dot{z_4}}{(1+z_4^2)^2(k_1^2+k_2^2+k_3^2)^{3/2}}\\
&+\dfrac{3\Bigl(2z_4k_1-(1-z_4^2)k_2 \Bigl)\Bigl( k_1\dot{k_1}+k_2\dot{k_2}+k_3\dot{k_3}\Bigl)^2}{(1+z_4^2)(k_1^2+k_2^2+k_3^2)^{5/2}}\Bigl)\\
&-\dfrac{1}{2}\dfrac{1}{\Bigl( 1-\dfrac{(2z_4k_1-(1-z_4^2)k_2)^2}{(1+z_4^2)^2(k_1^2+k_2^2+k_3^2)} \Bigl)^{3/2}} \Bigl( \Bigl( \dfrac{2\dot{z_4}k_1+2z_4\dot{k_1}+2z_4\dot{z_4}k_2-(1-z_4^2)\dot{k_2}}{(1+z_4^2)\sqrt{k_1^2+k_2^2+k_3^2}}\\
&-\dfrac{2\Bigl(2z_4k_1-(1-z_4^2)k_2 \Bigl) z_4 \dot{z_4}}{(1+z_4^2)^2\sqrt{k_1^2+k_2^2+k_3^2}}-\dfrac{\Bigl(2z_4k_1-(1-z_4^2)k_2 \Bigl)\Bigl( k_1\dot{k_1}+k_2\dot{k_2}+k_3\dot{k_3}\Bigl)}{(1+z_4^2)(k_1^2+k_2^2+k_3^2)^{3/2}}\Bigl)\\
&\Big(\dfrac{2\Bigl( 2z_4k_1-(1-z_4^2)k_2 \Bigl)\Big( 2\dot{z_4}k_1+2z_4\dot{k_1}+2z_4\dot{z_4}k_2-(1-z_4^2)\dot{k_2} \Bigl)}{(1+z_4^2)^2(k_1^2+k_2^2+k_3^2)}\\
&+\dfrac{4\Bigl( 2z_4k_1-(1-z_4^2)k_2 \Bigl)^2z_4 \dot{z_4}}{(1+z_4^2)^3(k_1^2+k_2^2+k_3^2)}+\dfrac{\Bigl( 2z_4k_1-(1-z_4^2)k_2 \Bigl)^2\Bigl( 2k_1\dot{k_1}+2k_2\dot{k_2}+2k_3\dot{k_3} \Bigl)}{(1+z_4^2)^2(k_1^2+k_2^2+k_3^2)^2} \Bigl) \Bigl)\\
&-\dfrac{2\Bigl( -2z_4 \dot{z_4} k_1 +(1-z_4^2)\dot{k_1}+2\dot{z_4}k_2+2z_4\dot{k_2}\Bigl) \dot{z_4}}{k_3(1+z_4^2)^2\sqrt{1+\dfrac{\Bigl((1-z_4^2)k_1+2z_4k_2 \Bigl)^2}{k_3^2(1+z_4^2)^2}}}+\dfrac{2\Bigl((1-z_4^2)k_1+2z_4k_2\Bigl)\dot{z_4}\dot{k_3}}{k_3^2(1+z_4^2)^2\sqrt{1+\dfrac{\Bigl((1-z_4^2)k_1+2z_4k_2 \Bigl)^2}{k_3^2(1+z_4^2)^2}}}
\end{split}
\end{equation}
\begin{equation}
\begin{split}
&+\dfrac{8\Bigl( (1-z_4^2)k_1+2z_4k_2 \Bigl)z_4 \dot{z_4}^2}{k_3^2(1+z_4^2)^3\sqrt{1+\dfrac{\Bigl((1-z_4^2)k_1+2z_4k_2 \Bigl)^2}{k_3^2(1+z_4^2)^2}}} + \dfrac{1}{k_3^2(1+z_4^2)^2 \left(1+\dfrac{\Bigl((1-z_4^2)k_1+2z_4k_2 \Bigl)^2}{k_3^2(1+z_4^2)^2}\right)^{3/2}}\Bigl(\\
&\dot{z_4}\Bigl((1-z_4^2)k_1+2z_4k_2\Bigl)\Bigl(\dfrac{2\left( (1-z_4^2)k_1+2z_4k_2\right)\left( -2z_4\dot{z_4}k_1+(1-z_4^2)\dot{k_1}+2\dot{z_4}k_2+2z_4\dot{k_2} \right)}{k_3^2(1+z_4^2)^2}\\
&-\dfrac{2\left( (1-z_4^2)k_1+2z_4k_2 \right)^2\dot{k_3}}{k_3^3(1+z_4^2)^2} - \dfrac{4z_4\dot{z_4}\left( (1-z_4^2)k_1+2z_4k_2  \right)^2}{k_3^2(1+z_4^2)^3}\Bigl) \Bigl) \\
&-\dfrac{2\ddot{z_4}\left( (1-z_4^2)k_1+2z_4k_2\right)}{k_3(1+z_4^2)^2\sqrt{1+\dfrac{\Bigl((1-z_4^2)k_1+2z_4k_2 \Bigl)^2}{k_3^2(1+z_4^2)^2}}}+(I_{yy}-I_{zz})\Bigl(\dfrac{1}{1+\dfrac{\Bigl((1-z_4^2)k_1+2z_4k_2 \Bigl)^2}{k_3^2(1+z_4^2)^2}}\Bigl(\\
&\sqrt{1-\dfrac{(2z_4k_1-(1-z_4^2)k_2)^2}{(1+z_4^2)^2(k_1^2+k_2^2+k_3^2)}}\Bigl( \dfrac{1}{k_3(1+z_4^2)}\Bigl(-2z_4\dot{z_4}k_1+(1-z_4^2)\dot{k_1}+2\dot{z_4}k_2+2z_4\dot{k_2}\Bigl)\\
&-\dfrac{\dot{k_3}\left( (1-z_4^2)k_1+2z_4k_2 \right)}{k_3^2(1+z_4^2)}-\dfrac{2z_4\dot{z_4}((1-z_4^2)k_1+2z_4k_2)}{k_3(1+z_4^2)^2}\Bigl) \Bigl)\\
&+\dfrac{2\left( 2z_4k_1-(1-z_4^2)k_2\right)\dot{z_4}}{(1+z_4^2)^2\sqrt{k_1^2+k_2^2+k_3^2}\sqrt{1+\dfrac{\Bigl((1-z_4^2)k_1+2z_4k_2 \Bigl)^2}{k_3^2(1+z_4^2)^2}}} \Bigl) \Bigl(\\
&-\dfrac{1}{(1+z_4^2)\sqrt{k_1^2+k_2^2+k_3^2}\sqrt{1+\dfrac{\Bigl((1-z_4^2)k_1+2z_4k_2 \Bigl)^2}{k_3^2(1+z_4^2)^2}}}\Bigl( \Bigl( 2z_4k_1-(1-z_4^2)k_2 \Bigl) \Bigl( \\
&\dfrac{-2z_4\dot{z_4}k_1+(1-z_4^2)\dot{k_1}+2\dot{z_4}k_2+2z_4\dot{k_2}}{k_3(1+z_4^2)}-\dfrac{\left((1-z_4^2)k_1+2z_4k_2\right)\dot{z_3}}{k_3^2(1+z_4^2)}\\
&-\dfrac{2\left((1-z_4^2)k_1+2z_4k_2 \right)z_4\dot{z_4}}{k_3(1+z_4^2)^2} \Bigl) \Bigl) +\dfrac{2\dot{z_4}\sqrt{1-\dfrac{(2z_4k_1-(1-z_4^2)k_2)^2}{(1+z_4^2)^2(k_1^2+k_2^2+k_3^2)}}}{(1+z_4^2)\sqrt{1+\dfrac{\left( (1-z_4^2)k_1+2z_4k_2\right)^2}{k_3^2(1+z_4^2)^2}}}
\end{split}
\end{equation}
\begin{equation}
\nonumber
\begin{split}
\tau_\theta&=I_{yy}\Big(\dfrac{1}{2}\Bigl( \Bigl( \dfrac{-2z_4\dot{z_4}k_1+(1-z_4^2)\dot{k_1}+2\dot{z_4}k_2+2z_4\dot{k_2}}{k_3(1+z_4^2)}-\dfrac{\left( (1-z_4^2)k_1+2z_4k_2 \right)\dot{k_3}}{k_3^2(1+z_4^2)}\\
&-\dfrac{2z_4\dot{z_4}\left( (1-z_4^2)k_1+2z_4k_2  \right)}{k_3(1+z_4^2)^2}\Bigl)\Bigl(-\dfrac{2\left(2z_4k_1 -(1-z_4^2)k_2 \right)\left( 2\dot{z_4}k_1+2z_4\dot{k_1}+2z_4\dot{z_4}k_2-(1-z_4^2)\dot{k_2} \right)}{(1+z_4^2)^2(k_1^2+k_2^2+k_3^2)}\\
&+\dfrac{4z_4\dot{z_4}\left(2z_4k_1 -(1-z_4^2)k_2 \right)^2}{(1+z_4^2)^3(k_1^2+k_2^2+k_3^2)}+\dfrac{2k_1\dot{k_1}\left( 2z_4 k_1 -(1-z_4^2)k_2 \right)^2+2 k_2\dot{k_2}
+2 k_3\dot{k_3}}{(1+z_4^2)^2(k_1^2+k_2^2+k_3^2)^2} \Bigl) \Bigl) \Bigl/ \\
&\left( \sqrt{1-\dfrac{(2z_4k_1-(1-z_4^2)k_2)^2}{(1+z_4^2)^2(k_1^2+k_2^2+k_3^2)}} \left( 1+\dfrac{\left( (1-z_4^2)k_1+2z_4k_2\right)^2}{k_3^2(1+z_4^2)^2} \right) \right)+\dfrac{1}{1+\dfrac{\left( (1-z_4^2)k_1+2z_4k_2\right)^2}{k_3^2(1+z_4^2)^2}}\Bigl(\\
&\sqrt{1-\dfrac{(2z_4k_1-(1-z_4^2)k_2)^2}{(1+z_4^2)^2(k_1^2+k_2^2+k_3^2)}}\Bigl( \dfrac{4\left( (1-z_4^2)k_1+2z_4k_2 \right)\dot{k_3}z_4\dot{z_4} -\left( (1-z_4^2)k_1+2z_4k_2 \right)\dot{k_3}^2}{k_3^2(1+z_4^2)^2}\\
&+\dfrac{-2\dot{z_4}^2k_1-2z_4\ddot{z_4}k_1-4z_4 \dot{z_4} \dot{z_1} + (1-z_4^2)\ddot{k_1} +2 \ddot{z_4}k_2+4\dot{z_4}\dot{k_2}+2z_4\ddot{k_2}}{k_3(1+z_4^2)} +\dfrac{8\left( (1-z_4^2)k_1+2z_4k_2 \right)z_4^2\dot{z_4}^2}{k_3(1+z_4^2)^3}\\
&-\dfrac{2\dot{k_3}\Bigl( -2z_4\dot{z_4}k_1+(1-z_4^2)\dot{k_1}+2\dot{z_4}k_2+2z_4\dot{k_2} \Bigl)}{k_3^2(1+z_4^2)} +\dfrac{2\left( (1-z_4^2)k_1+2z_4k_2 \right) \dot{k_3}^2}{k_3^3(1+z_4^2)}\\
&-\dfrac{4\left( -2z_4\dot{z_4}k_1+(1-z_4^2)\dot{k_1}+2\dot{z_4}k_2+2z_4\dot{k_2} \right)z_4\dot{z_4}+2\left( (1-z_4^2)k_1+2z_4k_2  \right) \left( \dot{z_4}^2 + z_4\ddot{z_4}\right)}{k_3(1+z_4^2)^2}\Bigl) \Bigl)\\
&-\dfrac{1}{\left( 1+ \dfrac{\left( 1-z_4^2)k_1+2z_4k_2 \right)^2}{k_3^2(1+z_4^2)^2}\right)^2}\Bigl( \sqrt{1-\dfrac{(2z_4k_1-(1-z_4^2)k_2)^2}{(1+z_4^2)^2(k_1^2+k_2^2+k_3^2)}}\Bigl( -\dfrac{\left((1-z_4^2)k_1+2z_4k_2  \right)\dot{k_3}}{k_3^2(1+z_4^2)}\\
&+\dfrac{-2z_4\dot{z_4}k_1+(1-z_4^2)\dot{k_1}+2\dot{z_4}k_2+2z_4\dot{k_2}}{k_3(1+z_4^2)}-\dfrac{2\left( (1-z_4^2)k_1+2z_4k_2 \right)z_4 \dot{z_4}}{k_3(1+z_4^2)^2} \Bigl)\\
&\Bigl( \dfrac{2\left( (1-z_4^2)k_1+2z_4k_2 \right)\left( -2z_4\dot{z_4}k_1+(1-z_4^2)\dot{k_1}+2\dot{z_4}k_2+2z_4\dot{k_2} \right)}{k_3^2(1+z_4^2)^2} - \dfrac{2\left(1-z_4^2)k_1+2z_4k_2 \right)^2\dot{k_3}}{k_3^3(1+z_4^2)^2}\\
&-\dfrac{4\left((1-z_4^2)k_1+2z_4k_2 \right)^2z_4\dot{z_4}}{k_3^2(1+z_4^2)^3}\Bigl) \Bigl) + \dfrac{2\left( 2\dot{z_4}k_1+2z_4\dot{k_1}+2z_4\dot{z_4}k_2-(1-z_4^2)\dot{k_2} \right)\dot{z_4}}{(1+z_4^2)^2\sqrt{k_1^2+k_2^2+k_3^2}\sqrt{1+\dfrac{\Bigl((1-z_4^2)k_1+2z_4k_2 \Bigl)^2}{k_3^2(1+z_4^2)^2}}}\\
&-\dfrac{8\left( 2z_4k_1-(1-z_4^2)k_2 \right)\dot{z_4}^2z_4}{(1+z_4^2)^3\sqrt{k_1^2+k_2^2+k_3^2}\sqrt{1+\dfrac{\Bigl((1-z_4^2)k_1+2z_4k_2 \Bigl)^2}{k_3^2(1+z_4^2)^2}}}
\end{split}
\end{equation}
\begin{equation}
\begin{split}
&-\dfrac{\dot{z_4}\left( 2z_4k_1-(1-z_4^2)k_2 \right)\left(2k_1\dot{k_1}+2k_2\dot{k_2}+2k_3\dot{k_3} \right)}{(1+z_4^2)^2(k_1^2+k_2^2+k_3^2)^{3/2}\sqrt{1+\dfrac{\Bigl((1-z_4^2)k_1+2z_4k_2 \Bigl)^2}{k_3^2(1+z_4^2)^2}}}\\
&-\dfrac{1}{(1+z_4^2)^2(k_1^2+k_2^2+k_3^2)^{3/2}\left(1+\dfrac{\Bigl((1-z_4^2)k_1+2z_4k_2 \Bigl)^2}{k_3^2(1+z_4^2)^2}\right)^{3/2}}\Bigl( \left(  2z_4k_1-(1-z_4^2)k_2 \right) \dot{z_4}\\
&\Bigl( \dfrac{2\left((1-z_4^2)k_1+2z_4k_2 \right) \left( -2z_4\dot{z_4}k_1+(1-z_4^2)\dot{k_1}+2\dot{z_4}k_2+2z_4\dot{k_2} \right)}{k_3^2(1+z_4^2)^2}-\dfrac{2\left( (1-z_4^2)k_1+2z_4k_2 \right)^2\dot{k_3}}{k_3^2(1+z_4^2)^2}\\
&-\dfrac{4\left( (1-z_4^2)k_1+2z_4k_2 \right)^2 z_4 \dot{z_4}}{k_3^2(1+z_4^2)^3}\Bigl) \Bigl) + \dfrac{2\left( 2z_4k_1-(1-z_4^2)k_2  \right)\ddot{z_4}}{(1+z_4^2)^2(k_1^2+k_2^2+k_3^2)^{3/2}\sqrt{1+\dfrac{\Bigl((1-z_4^2)k_1+2z_4k_2 \Bigl)^2}{k_3^2(1+z_4^2)^2}}}\Bigl)\\
&-(I_{zz}-I_{xx})\Bigl(\dfrac{1}{\sqrt{1-\dfrac{(2z_4k_1-(1-z_4^2)k_2)^2}{(1+z_4^2)^2(k_1^2+k_2^2+k_3^2)}}}\Bigl( \dfrac{2\dot{z_4}k_1+2z_4\dot{k_1}+2z_4\dot{z_4}k_2-(1-z_4^2)\dot{k_2}}{(1+z_4^2)\sqrt{k_1^2+k_2^2+k_3^2}}\\
&-\dfrac{2\left( 2z_4k_1-(1-z_4^2)k_2 \right)z_4\dot{z_4}}{(1+z_4^2)^2\sqrt{k_1^2+k_2^2+k_3^2}}-\dfrac{\left(2z_4k_1-(1-z_4^2)k_2 \right) \left( k_1\dot{k_1}+k_2\dot{k_2}+k_3\dot{k_3} \right)}{(1+z_4^2)(k_1^2+k_2^2+k_3^2)^{3/2}}\Bigl) \\
&-\dfrac{2\left( (1-z_4^2)k_1+2z_4k_2  \right)\dot{z_4}}{k_3(1+z_4^2)^2\sqrt{1+\dfrac{\Bigl((1-z_4^2)k_1+2z_4k_2 \Bigl)^2}{k_3^2(1+z_4^2)^2}}}\Bigl) \Bigl(-\dfrac{1}{(1+z_4^2)^2\sqrt{k_1^2+k_2^2+k_3^2}\left( 1+\dfrac{\Bigl((1-z_4^2)k_1+2z_4k_2 \Bigl)^2}{k_3^2(1+z_4^2)^2} \right)}\Bigl( \\
&\left( 2z_4k_1-(1-z_4^2)k_2 \right)\Bigl( \dfrac{-2z_4\dot{z_4}k_1+(1-z_4^2)\dot{k_1}+2\dot{z_4}k_2+2z_4\dot{k_2}}{k_3(1+z_4^2)}-\dfrac{((1-z_4^2)k_1+2z_4k_2)\dot{k_3}}{k_3^2(1+z_4^2)}\\
&-\dfrac{2((1-z_4^2)k_1+2z_4k_2)z_4\dot{z_4}}{k_3(1+z_4^2)^2}\Bigl) \Bigl)+\dfrac{2\dot{z_4}\sqrt{1-\dfrac{(2z_4k_1-(1-z_4^2)k_2)^2}{(1+z_4^2)^2(k_1^2+k_2^2+k_3^2)}}}{(1+z_4^2)\sqrt{1+\dfrac{\left( (1-z_4^2)k_1+2z_4k_2\right)^2}{k_3^2(1+z_4^2)^2}}}\Bigl)
\end{split}
\end{equation}
\begin{equation}
\nonumber
\begin{split}
\tau_\psi&=I_{zz}\Bigl(-\dfrac{1}{(1+z_4^2)^2\sqrt{k_1^2+k_2^2+k_3^2}\left(1+\dfrac{\Bigl((1-z_4^2)k_1+2z_4k_2 \Bigl)^2}{k_3^2(1+z_4^2)^2}\right)}\Bigl( \Bigl( 2\dot{z_4}k_1+2z_4\dot{k_1}+2z_4\dot{z_4}k_2\\
&-(1-z_4^2)\dot{k_2} \Bigl)\Bigl(\dfrac{-2z_4\dot{z_4}k_1+(1-z_4^2)\dot{k_1}+2\dot{z_4}k_2+2z_4\dot{k_2}}{k_3(1+z_4^2)}-\dfrac{((1-z_4^2)k_1+2z_4k_2 )\dot{k_3}}{k_3^2(1+z_4^2)}\\
&-\dfrac{2((1-z_4^2)k_1+2z_4k_2)z_4\dot{z_4}}{k_3(1+z_4^2)^2}\Bigl)z_4\dot{z_4}\Bigl)+\dfrac{1}{2}\dfrac{1}{(1+z_4^2)^2(k_1^2+k_2^2+k_3^2)^{3/2}\left(1+\dfrac{\Bigl((1-z_4^2)k_1+2z_4k_2 \Bigl)^2}{k_3^2(1+z_4^2)^2}\right)}\\
&\Bigl(\left(2z_4k_1-(1-z_4^2)k_2\right)\Bigl(\dfrac{-2z_4\dot{z_4}k_1+(1-z_4^2)\dot{k_1}+2\dot{z_4}k_2+2z_4\dot{k_2}}{k_3(1+z_4^2)}-\dfrac{((1-z_4^2)k_1+2z_4k_2)\dot{k_3}}{k_3^2(1+z_4^2)}\\
&-\dfrac{2\left( (1-z_4^2)k_1+2z_4k_2 \right)z_4\dot{z_4}}{k_3(1+z_4^2)^2}\Bigl)\left(2k_1\dot{k_1}+2k_2\dot{k_2}+2k_3\dot{k_3} \right)\Bigl)\\
&-\dfrac{1}{(1+z_4^2)\sqrt{k_1^2+k_2^2+k_3^2}\left(1+\dfrac{\Bigl((1-z_4^2)k_1+2z_4k_2 \Bigl)^2}{k_3^2(1+z_4^2)^2}\right)}\Bigl(\left( 2z_4k_1-(1-z_4^2)k_2 \right)\Bigl(\dfrac{1}{k_3(1+z_4^2}\Bigl(\\
&-2\dot{z_4}^2k_1-2z_4\ddot{z_4}k_1-4z_4\dot{z_4}\dot{k_1}+(1-z_4^2)\ddot{k_1}+2\ddot{z_4}k_2+4\dot{z_4}\dot{k_2}+2z_4\ddot{k_2}\Bigl)\\
&-\dfrac{2\left(-2z_4\dot{z_4}k_1+(1-z_4^2)\dot{k_1}+2\dot{z_4}k_2+2z_4\dot{k_2} \right)\dot{k_3}}{k_3^2(1+z_4^2)}-\dfrac{4\left(-2z_4\dot{z_4}k_1+(1-z_4^2)\dot{k_1}+2\dot{z_4}k_2+2z_4\dot{k_2} \right)z_4 \dot{z_4}}{k_3(1+z_4^2)^2}\\
&+\dfrac{2\left((1-z_4^2)k_1+2z_4k_2 \right)\dot{k_3}^2}{k_3^3(1+z_4^2)}+\dfrac{4\left(  (1-z_4^2)k_1+2z_4k_2 \right)\dot{k_3}z_4\dot{z_4}}{k_3^2(1+z_4^2)^2}-\dfrac{\left(  (1-z_4^2)k_1+2z_4k_2 \right)\ddot{k_3}}{k_3^2(1+z_4^2)}\\
&+\dfrac{8\left(  (1-z_4^2)k_1+2z_4k_2 \right)z_4^2\dot{z_4}^2}{k_3(1+z_4^2)^3}-\dfrac{2\left(  (1-z_4^2)k_1+2z_4k_2 \right)(\dot{z_4}^2+z_4\ddot{z_4})}{k_3(1+z_4^2)^2}\Bigl)\Bigl)\\
&+\dfrac{1}{(1+z_4^2)\sqrt{k_1^2+k_2^2+k_3^2}\left(1+\dfrac{\Bigl((1-z_4^2)k_1+2z_4k_2 \Bigl)^2}{k_3^2(1+z_4^2)^2}\right)^2}\Bigl(\left( 2z_4k_1-(1-z_4^2)k_2\right)\\
&\Bigl( \dfrac{-2z_4\dot{z_4}k_1+(1-z_4^2)\dot{k_1}+2\dot{z_4}k_2+2z_4\dot{k_2}}{k_3(1+z_4^2)}-\dfrac{\left(  (1-z_4^2)k_1+2z_4k_2 \right)\dot{k_3}}{k_3^2(1+z_4^2)}-\dfrac{2\left(  (1-z_4^2)k_1+2z_4k_2 \right)z_4\dot{z_4}}{k_3(1+z_4^2)^2}\Bigl)\\
&\Bigl(\dfrac{2\left(  (1-z_4^2)k_1+2z_4k_2 \right)\left( -2z_4\dot{z_4}k_1+(1-z_4^2)\dot{k_1}+2\dot{z_4}k_2+2z_4\dot{k_2}\right)}{k_3^2(1+z_4^2)^2}-\dfrac{2\left(  (1-z_4^2)k_1+2z_4k_2 \right)^2\dot{k_3}}{k_3^3(1+z_4^2)^2}\\
&-\dfrac{4\left(  (1-z_4^2)k_1+2z_4k_2 \right)^2z+4\dot{z_4}}{k_3^2(1+z_4^2)^3}\Bigl) \Bigl)+\Bigl(\dot{z_4}\Bigl(-\dfrac{2(2z_4k_1-(1-z_4^2)k_2)(2\dot{z_4}k_1+2z_4\dot{k_1}+2z_4\dot{z_4}k_2-(1-z_4^2)\dot{k_2})}{(1+z_4^2)^2(k_1^2+k_2^2+k_3^2)}
\end{split}
\end{equation}
\begin{equation}
\begin{split}
&+\dfrac{4(2z_4k_1-(1-z_4^2)k_2)^2z_4\dot{z_4}}{(1+z_4^2)^3(k_1^2+k_2^2+k_3^2)}+\dfrac{(2z_4k_1-(1-z_4^2)k_2)^2(2k_1\dot{k_1}+2k_2\dot{k_2}+2k_3\dot{k_3})}{(1+z_4^2)^2(k_1^2+k_2^2+k_3^2)^2}\Bigl) \Bigl) \Bigl/ \\
&\left(\sqrt{1-\dfrac{(2z_4k_1-(1-z_4^2)k_2)^2}{(1+z_4^2)^2(k_1^2+k_2^2+k_3^2)}} \sqrt{1+\dfrac{\left( (1-z_4^2)k_1+2z_4k_2\right)^2}{k_3^2(1+z_4^2)^2}} (1+z_4^2) \right)\\
&-\dfrac{1}{\left(1+\dfrac{\left( (1-z_4^2)k_1+2z_4k_2\right)^2}{k_3^2(1+z_4^2)^2}\right)^{3/2} (1+z_4^2)} \Bigl( \dot{z_4}\sqrt{1-\dfrac{(2z_4k_1-(1-z_4^2)k_2)^2}{(1+z_4^2)^2(k_1^2+k_2^2+k_3^2)}}\\
&\Bigl(\dfrac{2\left( (1-z_4^2)k_1+2z_4k_2\right)\left( -2z_4\dot{z_4}k_1+(1-z_4^2)\dot{k_1}+2\dot{z_4}k_2+2z_4\dot{k_2}\right)}{k_3^2(1+z_4^2)^2}\Bigl)\\
&-\dfrac{2\left( (1-z_4^2)k_1+2z_4k_2\right)^2\dot{k_3}}{k_3^3(1+z_4^2)^2}-\dfrac{4\left( (1-z_4^2)k_1+2z_4k_2\right)z_4\dot{z_4}}{k_3^2(1+z_4^2)^3} \Bigl) \Bigl)+\dfrac{2\ddot{z_4}\sqrt{1-\dfrac{(2z_4k_1-(1-z_4^2)k_2)^2}{(1+z_4^2)^2(k_1^2+k_2^2+k_3^2)}}}{(1+z_4^2)\sqrt{1+\dfrac{\left( (1-z_4^2)k_1+2z_4k_2\right)^2}{k_3^2(1+z_4^2)^2}}}\\
&-\dfrac{4\dot{z_4}^2z_4\sqrt{1-\dfrac{(2z_4k_1-(1-z_4^2)k_2)^2}{(1+z_4^2)^2(k_1^2+k_2^2+k_3^2)}}}{(1+z_4^2)^2\sqrt{1+\dfrac{\left( (1-z_4^2)k_1+2z_4k_2\right)^2}{k_3^2(1+z_4^2)^2}}}\Bigl)
\end{split}
\end{equation}
\end{document}